\documentclass[12pt]{amsart}
\usepackage{amssymb,latexsym}
\usepackage{enumerate}
\usepackage{graphicx}
\usepackage{psfrag}
\makeatletter
\@namedef{subjclassname@2010}{%
\textup{2010} Mathematics Subject Classification}
\makeatother

\newtheorem{thm}{Theorem}[section]

\theoremstyle{definition}
\newtheorem{defin}[thm]{Definition}
\newtheorem{rem}[thm]{Remark}

\numberwithin{equation}{section}

\def \lim   {\text {\rm lim}}

\def \tr    {\text {\rm tr}}

\begin{document}


\baselineskip=17pt



\title[Invariants for maximally entangled vectors \& unitary bases]{Invariants
for maximally entangled vectors and unitary bases}

\author[Sibasish Ghosh]{Sibasish Ghosh}
 \address{Optics \& Quantum Information Group\\The Institute of Mathematical
Sciences\\C.I.T. Campus, Taramani\\ Chennai-600113, India}
 \email{sibasish@imsc.res.in}

\author[Ajit Iqbal Singh]{Ajit Iqbal Singh}
\address{Theoretical Statistics and Mathematics Unit\\ Indian
Statistical Institute\\ 7, S. J. S. Sansanwal Marg\\ New Delhi-110 016,
India}
\email{aisingh@isid.ac.in, aisingh@sify.com}

\date{}

\begin{abstract}
The purpose of this paper is to study the equivalence relation on unitary bases
defined by R. F. Werner [{\it J. Phys. A: Math. Gen.} {\bf 34} (2001) 7081], relate 
it to local operations on maximally entangled vectors bases, find an invariant for 
equivalence classes in terms of certain commuting systems, and, relate it to 
mutually unbiased bases and Hadamard matrices. Illustrations are given in the 
context of latin squares and projective representations as well. Applications to 
quantum tomography are indicated.
\end{abstract}

\subjclass[2010]{Primary 81P40; Secondary 05B20, 15B34; 20C25}

\keywords{}

\maketitle

\section{Introduction}
The well-known Bell basis \cite{Be} for two qubits can be utilised to perform
different roles like to construct (i) an orthonormal basis of maximally
entangled vectors (MEV) or states (MES)  in $\mathbb{C}^2 \otimes \mathbb{C}^2,$
 (ii) an orthonormal basis of unitary operators $(UB)$ on $\mathbb{C}^2,$  (iii)
a complete system of mutually unbiased bases (MUB) in   $\mathbb{C}^2,$ (iv) a
prototype for MEV's in the sense that all MEV's can be obtained from Bell basis
by local operations.
\vskip0.1in
The roles in (i) and (ii) were achieved by the generalization to two qudits
(known as Bell states again) by Bennett et. al \cite{BBCJPW}.  Knill \cite{EK}
gave a method to construct unitary bases, called {\it nice error bases} in terms
of projective unitary representations of a finite group. The most basic nice error 
bases in terms of Weyl operators have been vigorously pursued amongst others, by
Parthasarathy and associates; an elementary account is available in \cite{Pa}.
It was figured out by Vollbrecht and Werner \cite{VW} that Bell states do not
perform other roles or have some other properties like (iv) of Bell basis qubits
if $d > 2.$ Werner \cite{Wer} established a one-to-one correspondence between
``tight'' quantum teleportation and dense coding schemes and also MEV, UB and
depolarizing operations. He gave a general construction procedure for UB
called {\it shift and multiply} involving latin squares and Hadamard matrices,
which generalises Weyl operators again. He
also defined an equivalence relation, say, $\sim$  between UB's. Ivanovic
\cite{Iv}, Wootters and Fields \cite{WF}, Bandyopadhyay et al. \cite{BBRV},
Lawrence et al. \cite{LBZ} constructed MUB's for $d,$ a prime and a prime power,
gave a relationship between MUB's and commuting properties of a decomposition of
UB, and pointed out the limitations of UB, for certain composite $d,$ to give
rise to a complete system of MUB's as in (iii) above.  The starting point
 for all of them was to determine a quantum state using quantum measurements
that correspond to pure states arising from the sought after complete system of
MUB's.
\vskip0.1in
The purpose of this paper is to study $\sim,$ relate it to local operations on
MEV's, find an invariant in terms of certain commuting systems for equivalence
classes and relate them to MUB's, underlying projective representations, or, 
latin squares and Hadamard matrices, and, to measurements of a quantum
state. 

In Section 2 of the paper, we introduce fan representations of unitary bases. The concepts and results of Section 2 are illustrated via different examples in Section 3. Some of the ideas developed in Section 2 as well as the examples discussed in Section 3 are then applied in Section 4 to discuss the issue of quantum state tomography. Finally, we draw the conclusion in Section 5.   

\def\cf{{\mathcal{F}}}
\def\cF{{\mathcal{F}}}

\section{Unitary bases and their fan representation}

\subsection{Basics} We begin with some basic material.
\vskip0.1in
We shall freely use \cite{VW} and \cite{Wer} in this section. Let $\mathcal{H}$
be a $d$-dimensional Hilbert space with $2\leq d< \infty$ and $X$ a set of $d^2$
elements such as $\{1,2,\ldots,d^2 \}$ or $\{0,1,\ldots, d^2-1\}$ etc. .
\begin{itemize}
 \item[(i)] Let $\mathbf{U} = \{U_x : x \in X\}$ be a {\it unitary basis,} in
short, UB, i.e. a collection $\mathbf{U}$ of unitary operators $U_x \in
\mathcal{B}(\mathcal{H}),$ the $\ast$-algebra of bounded linear operators on
$\mathcal{H}$ to itself, such that $\tr (U_x^{\ast}U_y)=d \delta_{xy}$ for $x,y
\in X.$ 

\item[(ii)] Rewording a part of the discussion after Proposition 9 \cite{Wer},
we
call two unitary bases $\mathbf{U}$ and $\mathbf{U}^{\prime}$  {\it equivalent}
if there exist unitaries $V_1,$ $V_2$ in $\mathcal{B}(\mathcal{H})$ and a
relabelling $x \rightarrow x^{\prime}$ of $X$ such that $U_{x^{\prime}}^{\prime}
= V_1 U_x V_2$ for $x$ in $X.$

\item[(iii)] We fix an orthonormal basis $\mathbf{e} = \{e_j: 1 \leq j \leq d\}$
in $\mathcal{H}$ and identify $\mathcal{H}$ with $\mathbb{C}^d$ as such. This
will permit us to identify $\mathcal{B}(\mathcal{H})$ with the $\ast$-algebra
$M_d$ of complex $d \times d$ matrices.  We set
$\Omega = \frac{1}{\sqrt{d}} \sum\limits_j e_j \otimes e_j.$

\item[(iv)] We recall a  one-to-one linear correspondence $\tau$ (in terms of
this
maximally entangled vector $\Omega$ in $\mathcal{H} \otimes \mathcal{H}$)
between $\Psi \in \mathcal{H} \otimes \mathcal{H}$  and $A \in
\mathcal{B}(\mathcal{H}),$ as set up in \cite{VW}, Lemma 2 or \cite{Wer}, Proof
of Theorem 1, for instance, via $\langle e_j| Ae_{k}\rangle = \sqrt{d} \langle
e_j \otimes e_k, \Psi \rangle.$ At times we express this as $\Psi \rightarrow
A_{\Psi}$ or, even $A \rightarrow \Psi_A$ and write $\tau (\Psi)=A_{\Psi}.$

\item[(v)] Let $\tau$ be as in (iv) above.
\end{itemize}

\begin{itemize}
 \item[(a)] The map $\tau$ takes the set $\mathcal{M}$  of maximally entangled
vectors in $\mathcal{H} \otimes \mathcal{H}$ to the set
$\mathcal{U}(\mathcal{H})$ of unitaries on $\mathcal{H}.$
\item[(b)] The rank of $A$ and the Schmidt rank of $\Psi_A$ are the same.
\item[(c)] Entropy of $A^{\ast} A$ or certain variants are in vogue as measures
of entanglement of $\Psi_A.$
\end{itemize}
\begin{itemize}
\item[(vi)] The {\it flip} (or {\it swap}) operation in $\mathcal{H} \otimes
\mathcal{H}$ is the linear operator $\mathbb{F}$ determined by $\mathbb{F}
(\varphi \otimes \psi) = \psi \otimes \varphi,$ for $\varphi,$ $\psi \in
\mathcal{H},$ or equivalently, by $\mathbb{F}(e_j \otimes e_k) = e_k \otimes
e_j$ for $1 \leq j,$ $k \leq d.$ Then $\mathbb{F}$ induces the {\it transpose
operator } $A \rightarrow A^t$ on $\mathcal{B}(\mathcal{H})$ to itself defined
in the basis $\mathbf{e}.$ We note the useful underlying fact: $\Psi = (A_{\Psi}
\otimes I ) \Omega = (I \otimes A_{\Psi}^{t}) \Omega,$ where $I$ is  the
identity operator on $\mathcal{H}.$

\item[(vii)] Members of the set $\mathcal{L} \mathcal{U} (\mathcal{H} \otimes
\mathcal{H})$ of local unitaries, viz., $\{U_1 \otimes U_2 : U_1, U_2 \in
\mathcal{U}(\mathcal{H})\}$ and $\mathbb{F}$ all take unit product vectors to
unit product vectors and Proposition 4 in \cite{VW} says that a unitary operator
$U$ on $\mathcal{H} \otimes \mathcal{H}$ satisfies $U \mathcal{M} \subset
\mathcal{M}$ if and only if $U$ is local up to a flip, i.e., there are unitaries
$U_1, U_2$ on $\mathcal{H}$ such that either $U = U_1 \otimes U_2$ or \linebreak
$U = (U_1 \otimes U_2)$ $\mathbb{F}.$
\end{itemize}

\begin{defin}[Unitary systems]\label{def2.2} 
\begin{itemize}
\item[]
 \item[(i)] A non-empty set $\mathbf{W} = \{ W_y : y \in Y\}$ of unitaries in
$\mathcal{H}$ indexed by $Y$ will be called a {\it unitary system,} in short,
US, if $\tr \,W_y =0$ and $\tr \, (W_x^{\ast} W_y)=d \delta_{xy}$ for $x,$ $y$
in $Y.$ The number $s = \# Y$ will be called the {\it size} of $\mathbf{W}.$

\item[(ii)] An {\it abelian unitary system,} in short, {\it AUS}, is a unitary
system $\mathbf{W}$ with $W_xW_y=W_y W_x$ for $x,y$ in $Y.$

\item[(iii)] {\it A maximal abelian subsystem of a unitary system} $\mathbf{W},$
in short, $\mathbf{W}$-MASS, is a subset $\mathbf{V}$ of $\mathbf{W}$ which is
an AUS and maximal with this property.

\item[(iv)] A {\it tagged unitary system,} in short, TUS, is a triple
$\mathbf{T}=(x_0, U_{x_{0}}, \mathbf{W})$ where $x_0 \in X,$ $U_{x_{0}} \in
\mathcal{U} (\mathcal{H}),$ $\mathbf{W} = \{W_y : y \in X, y \neq x_0\}$ is a
unitary system. We say $\mathbf{T}$ is tagged at $x_0.$
\end{itemize}
\end{defin}

\begin{rem}\label{re2.2} 
Let $\mathbf{T}$ be a tag at $x_0.$ Set $W_{x_{0}} = I.$  We set
$\mathbf{U}=\{U_x = U_{x_{0}} W_x: x \in X\}.$ Then $\mathbf{U}$ satisfies
$W_x^{\ast} W_y = U_x^{\ast} U_y$ for $x,$ $y \in X,$ and, therefore,
$\mathbf{U}$ is a UB. We call $\mathbf{U}$ the {\it $\mathbf{T}$-associated UB.}
 On the other hand, given a UB $\mathbf{U},$ for $x_0 \in X,$ setting
$\mathbf{W} = \{W_x = U_{x_{0}}^{\ast} U_x, x \in X, x \neq x_0   \},$ we have
$\mathbf{T}= (x_0, U_{x_{0}}, \mathbf{W})$ is a tag at $x_0.$ It satisfies
$W_x^{\ast} W_y = U_x^{\ast} U_y$ for $x,$ $y \in X,$ where $W_{x_{o}}$ has been
taken as $I.$ We call $\mathbf{T}$ the $\mathbf{U}$-{\it associated tag} at
$x_0.$ We note that in both cases, for $x, y \in X,$ $W_y W_x^{\ast} =
U_{x_{0}}^{\ast} U_y U_x^{\ast} U_{x_{0}}.$ Further, for $x,$ $y \in X,$ $W_x
W_y = W_y W_x$ if and only if $U_x U_{x_{0}}^{\ast} U_y= U_y U_{x_{0}}^{\ast}
U_x.$ We denote this condition by $\mathbf{T}(x, x_0, y)$ and call it by {\it
Twill}. Really speaking, both the
associations are on the left and have obvious right versions as well. 
\end{rem}

\begin{rem}[Relationship with MUB and Hadamard matrices]\label{rem2.3} 
\begin{itemize}
  \item[]

  \item[(i)] Let $\mathbf{W}$ be a unitary system. Then $\tr\, A=0$ for each $A$
in the linear span $\mathcal{L}$ of $\mathbf{W}.$ So $I \not\in \mathcal{L}.$
Also $\tr\, (W_x^{\ast} W_y) = d \delta_{xy}$ for $x,y$ in $Y$ forces $\{W_x : x
\in Y\}$ to be linearly independent. So we have $s=\# Y \leq d^2-1$ and, thus,
we may consider $Y$ as a proper subset of $X,$ if we like.

\item[(ii)] Given a US $\mathbf{W},$ $\widetilde{\mathbf{W}} = \mathbf{W} \cup 
\{I\}$ will be called the {\it unitization} of $\mathbf{W}.$ We note that
$\widetilde{\mathbf{W}}$ is a system of mutually orthogonal unitaries and it is
a UB if and only if the size of $\mathbf{W}$ is $d^2-1.$

\item[(iii)] AUS of size $d-1$ have been very well utilized by Wootters and
Fields \cite{WF} and Bandyopadhyay et. al. \cite{BBRV}  to study {\it mutually
unbiased bases   } in $\mathcal{H},$ see also \cite{LBZ},
\cite{PDB}, \cite{Sz}, \cite{msp} and \cite{psmw}. (\cite{BBRV}, Lemma
3.1)
records the basic fact that the size of an AUS can at most be $d-1.$

\item[(iv)] In fact, if $\mathbf{W} = \{W_x: x \in Y\}$ is an AUS of size $s,$
then there is a unitary $U$ on $\mathcal{H}$ and mutually orthogonal operators
represented by mutually orthogonal diagonal matrices $D_x,$ $x \in Y$ with
entries in the unit circle $\mathbf{S}^1$ such that $\tr\, D_x=0,$ $W_x=U D_x
U^{\ast},$ $x \in Y.$ As a consequence,  the  $(s+1) \times d$ matrix $H$ formed
by the diagonals of $I$ and $D_x,$ $x \in Y,$ as rows is a partial complex
Hadamard matrix  in the sense that $HH^{\ast} = d I_{s+1}$  and $|H_{jk}|=1$ for
each entry $H_{jk}$ of $H.$ This  forces $s \leq d-1.$  History and development
of Hadamard matrices is very long and fascinating. We just mention a few sources
 \cite{deLL}, \cite{Ha}, \cite{PDB}, \cite{Sz} and \cite{TZ}, \cite{V1},
\cite{V2}, \cite{Ho} which we can directly use.

\item[(v)] (\cite{BBRV}, Theorem 3.2) says that if a $U B$ $\mathbf{U}$
containing $I$ can be partitioned as a union of $(d+1)$ $\mathbf{U}\backslash
\{I\}$-MASS's of size $(d-1)$ each together with $I,$ then one can construct a
complete system of $(d+1)$ MUB's of $\mathbb{C}^d.$ The converse is given by
them as (\cite{BBRV}, Theorem 3.4). Concrete illustration is given for $d=p^r,$
with $p$ a prime. 

\item[(vi)] Obstructions to construction of MUB have occupied many researchers,
e.g., see \cite{WF}, \cite{BBRV}, \cite{PDB}, \cite{msp}, \cite{psmw},
\cite{PPGB}. Even when, say for, $d$ a prime power, complete systems of MUB's
exist, some subsystems of certain UB's may not be extendable to  complete
systems of MUB's, this is explained well by Mandayam, Bandyopadhyay, Grassl and
Wootters \cite{psmw} for
$d=2^2, 2^3, \ldots.$ 
\end{itemize}
\end{rem}

Contents of the next three items are based on the fourth section of the second
author's preprint \cite{Si}.
 
\begin{defin}\label{def2.4} 
Let $\mathcal{F},$ $\mathcal{G}$ be subsets of $\mathcal{B}(\mathcal{H}).$
\begin{itemize}
 \item[(i)] We say that $\mathcal{F}$ is {\it collectively unitarily equivalent}
to $\mathcal{G},$ in short, $\mathcal{F}$CUE$\mathcal{G},$ if
there exists a
$V \in \mathcal{U}(\mathcal{H})$ such that $\mathcal{G} = \{V^{\ast} AV: A \in
\mathcal{F}\}.$ We may say $\mathcal{F}$CUE$\mathcal{G}$ via $V.$

\item[(ii)] In case $\mathcal{F}$ and $\mathcal{G}$ are decomposed as
$\mathcal{F} = \underset{\gamma \in \Gamma}{\cup} \mathcal{F}_{\gamma},$
$\mathcal{G} = \underset{\gamma \in \Gamma}{\cup} \mathcal{G}_{\gamma}$
respectively, then we may require that for $\gamma \in \Gamma,$
$\mathcal{F}_{\gamma}$CUE$\mathcal{G}_{\gamma}$ via $V,$ and say that
$\mathcal{F}_{\Gamma}$CUE$\mathcal{G}_{\Gamma},$ or, if no confusion arises
$\mathcal{F}$CUE$\mathcal{G}.$

\item[(iii)] In case $\mathcal{F}$ and $\mathcal{G}$ are indexed by a set
$\Lambda$ as $\{A_{\alpha}:\alpha \in \Lambda \}$ and $\{B_{\alpha} : \alpha \in
\Lambda \}$ respectively, then (as a special case of (ii) above) we may require
that $B_{\alpha} = V^{\ast} A_{\alpha} V, \alpha \in \Lambda$ and say that
$\mathcal{F}_{\Lambda}$CUE$\mathcal{G}_{\Lambda},$ or, if no confusion arises,
$\mathcal{F}$CUE$\mathcal{G}.$
\end{itemize}
\end{defin}

\begin{rem}\label{rem2.5} 
 \begin{itemize}
\item[]
  \item[(i)] The operator $V$ in the definition above may not be unique. In
fact, if $U^{\ast} \mathcal{G}U = \mathcal{G}$ for some $U \in
\mathcal{U}(\mathcal{H})$ then $VU$ works fine too.

\item[(ii)] $\mathcal{F}$CUE$\mathcal{G}$ if and only if
$\mathcal{F}_{\Lambda}$CUE$\mathcal{G}_{\Lambda}$ for some indexing
$\mathcal{F}_{\Lambda}$ and
$\mathcal{G}_{\Lambda}$ of $\mathcal{F}$ and $\mathcal{G}$ respectively by the
same index set $\Lambda.$ So we may fix some such indexing, if we like.

\item[(iii)] The relation CUE (in both the senses in Definition \ref{def2.4} 
above) is an equivalence relation.

\item[(iv)] It can be readily seen that if $\mathcal{F}$CUE$\mathcal{G}$ and
$\mathcal{F}$ is a commuting family then so is $\mathcal{G}.$

\item[(v)] If $\mathcal{F}_{\Lambda}$CUE$\mathcal{G}_{\Lambda}$ via $V,$ then
for each $\alpha \in \Lambda,$ the spectrum $\sigma (A_{\alpha})=\sigma
(B_{\alpha});$ and for $\alpha \in \Lambda,$ for an eigenvalue $\lambda$ of
$A_{\alpha},$ $\xi$ is an eigenvector for $B_{\alpha}$ with eigenvalue $\lambda$
if and only if $V\xi$ is an eigenvector for $A_{\alpha}$ with eigenvalue
$\lambda.$

\item[(vi)] If $\mathcal{F}= \{F_{\alpha}: 1 \leq \alpha \leq n\}$  is a
commuting $n$-tuple of normal operators in $\mathcal{B}(\mathcal{H}),$ then
there exists a $U \in \mathcal{U}(\mathcal{H})$ and an $n$-tuple $\mathcal{D}$
of operators ($D_{\alpha}: 1 \leq \alpha \leq n  $) represented by diagonal
matrices $\{\widetilde{D}_{\alpha} : 1 \leq \alpha \leq n\}$ with respect to
basis $\mathbf{e}$ such that $F_{\alpha} = U D_{\alpha} U^{\ast},$ $1 \leq
\alpha \leq n.$ In other words, $\mathcal{F}$CUE$\mathcal{D}.$ The  converse
is
also true.
 \end{itemize}
\end{rem}

\begin{rem}[{\it Application of CUE to construction of PPT
matrices}]\label{rem2.6} 
 \begin{itemize}
\item[]
  \item[(i)] Garcia and Tener (\cite{GT}, Theorem 1.1) obtained a canonical
decomposition for complex matrices $T$ which are UET,  i.e., unitarily
equivalent to their transposes $T^t,$ we may call a tuple $\mathcal{F} =
(F_{\alpha} : 1 \leq \alpha \leq n)$ of $d \times d$ matrices {\it collectively
unitarily equivalent} to the respective {\it transposes,} in short, CUET,
if
$ \mathcal{F}$CUE$\mathcal{F}^t,$ where $\mathcal{F}^t = (F_{\alpha}^t : 1 \leq
\alpha \leq n).$
\begin{itemize}
 \item[(a)]  Arveson \cite[Lemma
A.3.4]{Ar} gives that 
$$\left (\left (\begin{array}{ccc}0 & \lambda &
1
\\ 0 & 0 & 0 \\ 0 & 0 & 0 \end{array} \right ), \left
(\begin{array}{ccc} 0&0& \mu \\ 0 & 1 & 0 \\ 0  & -\lambda & 0
\end{array} \right ) \right )$$ is not CUET where $\lambda$ is a
non-real
complex number and $\mu$ is a complex number with $|\mu | = (1 +
|\lambda|^2)^{\frac{1}{2}}.$ 

\item[(b)] Garcia and Tener [\cite{GT}, expressions (1.4), (1.5), (1.6)] note that if
$T=\left (\begin{array}{cc} A & B \\ D & A^t \end{array} \right ),$ with $d=2n,$
$A,B,D$ $n \times n$ matrices satisfying $B^t = -B,$ $D^t = - D$ (to be
called skew-Hamiltonian (SHM, in short)) then it is UET simply because $T =
J T^t J^{\star},$ with $J=\left (\begin{array}{cc} 0 & I \\ -I & 0
\end{array} \right ).$  We can now immediately strengthen this to : for every 
$d,$ every non-empty collection  $\mathcal{F}$ of $d \times d$ SHM's is CUET via
$J.$ For more details one can see \cite{GT}, particularly \S 6 and \S7.

\item[(c)] (\cite{GT}, items 8.3, 8.4 and 8.5) tell us how to construct CUET
tuples.

\item[(d)] Tadez and Zyczkowski (\cite{TZ}, 4.1) indicate a general method to
construct a class of CUET tuples of matrices. Let $P$ be the permutation
matrix given in column form as $[e_1, e_d, e_{d-1}, \ldots, e_2].$ Let ${\xi} =
({\xi}_j) \in \mathbb{C}^d.$ Set $C_{\xi} = [C^{\xi}_{jk}]$ to be the $d \times
d$ circulant matrix given by $C^{\xi}_{jk} = \xi_{j-k}$ $(\mbox{mod}\, d).$ Then
$C_{\xi}^t = P^t C_{\xi} P.$ Because $P$ is a real unitary matrix, this gives
that any tuple consisting of $C_{\xi}$'s is CUET.
\end{itemize}

\item[(ii)] 

\noindent Let $[A_{jk}]$ be a positive block
matrix such that $(A_{jk}: 1 \leq j,k \leq n)$ is CUET. Then
$[A_{jk}^t]$  is positive. To see this we first note that there is a unitary
matrix $U$ such that
$A_{jk} = UA_{jk}^t U^{\ast}$ for $1 \leq j,k \leq n.$  Let
$\widetilde{U}$ be the block matrix $[\delta_{jk}
U],$ with $\delta_{jk} = 0$ for $j \neq k$ and $1$ for
$j=k.$ Then $\widetilde{U}$  is unitary and $[A_{jk}^t] =
\widetilde{U}^{\ast} [A_{jk}] \widetilde{U}.$ So
$[A_{jk}^t]$ is positive. 

\item[(iii)] Items (i)(b) and (ii) above can be put together to construct
matrices with positive partial transposes, in short, PPT matrices.
\vskip0.2in
\noindent{\bf Step 1:} Let $n \ge 2$ and put $m=\frac{n(n+1)}{2}.$
Use (i)(b) above to construct a CUET $m$-tuple $(Y_j:1 \leq j
\leq m)$ with $Y_j \ge 0$ for $1 \leq j \leq n$ (we may take all
$Y_j,$ $1 \leq j \leq n$ to be zero, for instance). For $1 \leq j
\leq m,$ we have $Y_j = QY_j^t Q^{\ast}$ and, therefore $Y_j^{\ast} =
QY_j^{\ast t} Q^{\ast}.$ We set $B_{jj} = Y_j$ for $1 \leq j \leq n,$
arrange $Y_j$ for $n+1 \leq j \leq \frac{n(n+1)}{2}$ as $B_{pq},$ $1
\leq p < q \leq n$ and take $B_{qp} = B_{pq}^{\ast}$ for $1 \leq p <
q \leq n.$ Thus, we obtain a block matrix $B = [B_{jk}]$ which is
Hermitian and $\{B_{jk} : 1 \leq j, k \leq n \}$ is CUET.
\vskip0.1in
\noindent{\bf Step 2:} The set $\{a
\in \mathbb{R}: B+a \,\, I_{n^{2}} \ge 0\}$ is an interval $[a_0,
\infty)$ for some $a_0 \in \mathbb{R}.$ We take any $a$ in this
interval and set $A=B+a \,\, I_{n^{2}}$ i.e., $A_{jk} = B_{jk}$ for $j
\neq k,$ whereas $A_{jj} = B_{jj} + a I_n$ for $1 \leq j, k \leq n.$
Then $\{A_{jk} : 1 \leq j, k \leq n \}$ is CUET and $A \ge 0.$ So we
can apply (ii) above to conclude that $A$ is a PPT matrix.

 \end{itemize}
\end{rem}

\begin{thm}\label{thm2.7} 
Let $\mathbf{U},$ $\mathbf{U}^{\prime}$ be unitary bases for $\mathcal{H}.$ Then
the following are equivalent.
\begin{itemize}
 \item[(i)] $\mathbf{U}$ is equivalent to $\mathbf{U}^{\prime},$
\item[(ii)] for some $\mathbf{U}$-associated tag $\mathbf{T}$ and some
$\mathbf{U}^{\prime}$-associated tag $\mathbf{T}^{\prime},$ \\ $\mathbf{W}$CUE
$\mathbf{W}^{\prime}.$

\item[(iii)] for each $\mathbf{U}$-associated tag $\mathbf{T},$ there is a
$\mathbf{U}^{\prime}$-associated tag $\mathbf{T}^{\prime}$ such that
$\mathbf{W}$CUE$\mathbf{W}^{\prime}.$
\end{itemize}
\end{thm}

\begin{proof}
 (i) $\Rightarrow$ (iii), Suppose $\mathbf{U} \sim \mathbf{U}^{\prime}.$ Then
there exist $V_1, V_2 \in \mathcal{U}(\mathcal{H})$ and a relabelling $x
\rightarrow x^{\prime}$ of $X$ such that $U_{x^{\prime}}^{\prime} = V_1 U_{x}
V_2$ for $x \in X.$ Consider any $x_0 \in X$ and let $\mathbf{T} = (x_0,
U_{x_{0}},
\mathbf{W})$ be the $\mathbf{U}$-associated tag at $x_0$ and
$\mathbf{T}^{\prime}= (x_0^{\prime}, U_{x_{0}^{\prime}}^{\prime},
\mathbf{W}^{\prime}),$ the $\mathbf{U}^{\prime}$-associated tag at
$x_0^{\prime}.$ Set $W_{x^{\prime}_{0}}^{\prime} = W_{x_{0}} = I. $  Then
$U_{x_{0}^{\prime}}^{\prime} = V_1 U_{x_{0}} V_2$ and, therefore, for $x \in X,$
$W_{x^{\prime}}^{\prime} =  U_{x_{0}^{\prime}}^{\prime^{\ast}}
U_{x^{\prime}}^{\prime}=V_2^{\ast} U_{x_{0}}^{\ast} V_1^{\ast} V_1 U_x V_2 =
V_2^{\ast} W_x V_2.$ So $\mathbf{W}$CUE$\mathbf{W}^{\prime}$ via $V_2.$
\vskip0.2in
(iii) $\Rightarrow$ (ii) is trivial.\\
(ii) $\Rightarrow$ (i), Let $\mathbf{T}  = (x_0, U_{x_{0}}, \mathbf{W})$ and 
$\mathbf{T}^{\prime} = (x_0^{\prime}, U_{x_{0}^{\prime}}^{\prime},
\mathbf{W}^{\prime})$ be the $U$-associated tag at $x_0$ and
$\mathbf{U}^{\prime}$-associated tag at $x_0^{\prime}$ respectively with
$\mathbf{W}$CUE$\mathbf{W}^{\prime}.$ Then by Remark \ref{rem2.5}(ii), there
exist a $V \in \mathcal{U}(\mathcal{H})$ and a bijective function on $X
\backslash \{x_0\}$ onto $X \backslash \{x_0^{\prime}\},$ say $x \rightarrow
x^{\prime}$ such that $W_{x^{\prime}}^{\prime} = V^{\ast} W_x V$  for $x \in X
\backslash \{x_0\}.$

Set $V_1 = U_{x_{0}^{\prime}}^{\prime} V^{\ast} U_{x_{0}}^{\ast}.$ Then $V_1 \in
\mathcal{U}(\mathcal{H}).$ Also $U_{x_{0}^{\prime}}^{\prime} = V_1 U_{x_{0}} V.$
Further, for $x \in X \backslash \{x_0\}$
\begin{eqnarray*}
U_{x^{\prime}}^{\prime} &=& U_{x_{0}^{\prime}}^{\prime} W_{x^{\prime}}^{\prime}
= U_{x_{0}^{\prime}}^{\prime} (V^{\ast} W_x V)\\ 
&=& (U_{x_{0}^{\prime}}^{\prime} V^{\ast} U_{x_{0}}^{\ast})(U_{x_{0}} W_{x})V   
      \\
&=& V_1 U_x V.
\end{eqnarray*}
So $\mathbf{U}$ is equivalent to $\mathbf{U}^{\prime}.$
\end{proof}

\begin{thm}\label{thm2.8}
Let $\mathbf{W}$ be a unitary system. 
\begin{itemize}
 \item[(i)] Then there is a unique maximal family
$\mathcal{V}_{\mathbf{W}}=\{\mathbf{V}_{\alpha}: \alpha \in \Lambda  \}$ of
$\mathbf{W}$-MASS's, such that $\mathbf{W}=\underset{\alpha \in \Lambda}{\cup}
\mathbf{V}_{\alpha}.$ If each $W_y$ in $\mathbf{W}$ has simple eigenvalues,
then $\mathbf{V}_{\alpha}$'s are mutually disjoint.

\item[(ii)] Let $\mathbf{W}^{\prime}$ be a unitary system and
$\mathcal{V}_{\mathbf{W}^{\prime}} = \{\mathbf{V}_{\beta}^{\prime}: \beta \in
\Lambda^{\prime}\},$ the family of $\mathbf{W}^{\prime}$-MASS's as in (i). Then
$\mathbf{W}$CUE$\mathbf{W}^{\prime}$ via $V$ iff there is a bijective map on
$\Lambda$ to $\Lambda^{\prime},$ say, $\alpha \rightarrow \alpha^{\prime},$ such
that $\mathbf{V}_{\alpha}$CUE$\mathbf{V}^{\prime}_{\alpha^{\prime}}$ via $V.$
\end{itemize}
\end{thm}

\begin{proof}
 \begin{itemize}
  \item[(i)]  Let $\mathbf{W} = \{W_y: y \in Y \}.$ For
any $y \in Y,$ $\{W_y\}$
is an AUS $\subset \mathbf{W}.$ Let $\mathcal{W} = \{\mathbf{A} : \mathbf{A}
\subset \mathbf{W} \,\mbox{and} \, \mathbf{A} \, \mbox{is an AUS} \} $ and order
it by inclusion. Then $\mathcal{V}_{\mathbf{W}}$ is made up of maximal elements
of $\mathcal{W}.$ The second part follows from elementary Linear Algebra as
indicated in Remark 2.9 (iv) below.

\item[(ii)] Suppose $\mathbf{W}$CUE$\mathbf{W}^{\prime}$ via $V.$ Then
$\{\mathbf{\widehat{V}_{\alpha}} =\{V^{\ast} AV : A \in \mathbf{V}_{\alpha} 
\}, \alpha
\in \Lambda \}$ is a  maximal family of $\mathbf{W}^{\prime}$-MASS's with
$\mathbf{W}^{\prime} =
\underset{\alpha \in \Lambda}{\cup} \widehat{\mathbf{V}}_{\alpha}.$ So, by
uniqueness, each $\widehat{\mathbf{V}}_{\alpha}$ is some unique
$\mathbf{V}_{\beta}^{\prime}.$ Set $\beta = \alpha^{\prime}.$ On the other hand,
each $\mathbf{V}_{\beta}^{\prime}$ is some unique
$\widehat{\mathbf{V}}_{\alpha}.$ So the map $\alpha \rightarrow
\alpha^{\prime}$ is bijective on $\Lambda$ to $\Lambda^{\prime}.$ The converse
part is trivial.
 \end{itemize}
\end{proof}


\begin{rem}\label{rem2.9}
\begin{itemize}
 \item []

\item [(i)] In view of Theorem \ref{thm2.7}, we may add a sixth condition to
Theorem 1 of \cite{Wer} viz., collection of tagged unitary systems as follows:\\
Tagged unitary systems $\mathbf{T},$ i.e., any arbitrarily fixed $x_0 \in X,$ a
unitary $U_{x_{0}} \in \mathcal{U} (\mathcal{H})$ and unitaries $\{W_x : x \in
X, x \neq x_0  \}$ such that $\tr \, W_x =0 = \tr \, (W_{x}^{\ast} W_y)$ for
$x,y, x \neq y$ in $X \backslash \{x_0\}.$

\item [(ii)] Theorem \ref{thm2.8} says that, to within CUE, we may think of
$\mathcal{V}_{\mathbf{W}}$ as an {\it invariant for a unitary system
$\mathbf{W}.$}

\item [(iii)] The role of Hadamard matrices in these
invariants has already been indicated above.  To elaborate a bit, for each
$\alpha \in \Lambda,$ there is a unitary $U_{\alpha}$ and a (partial) Hadamard
$s_{\alpha} \times d$ matrix $H_{\alpha}$ with $s_{\alpha} =
\,\,\mbox{size}\,\,\mathbf{V}_{\alpha}$ such that $\mathbf{V}_{\alpha}$ consists
of operators of the type $U_{\alpha} D U_{\alpha}^{\ast},$ $D$ is a diagonal
matrix whose diagonal forms a row of $H_{\alpha}.$ The ordering of rows
corresponds to that of operators in $\mathbf{V}_{\alpha}.$ To within that
$H_{\alpha}$ is unique upto a permutation of columns, and the corresponding
$U_{\alpha}$'s will undergo changes accordingly. For each $\alpha \in \Lambda,$
the augmented matrix $\widetilde{H}_{\alpha}$ formed by adding a top row of all
$1$'s is also a Hadamard matrix and it arises from the
$\widetilde{\mathbf{W}}$-MASS  $\widetilde{\mathbf{V}}_{\alpha} =
\mathbf{V}_{\alpha} \cup \{I\}$ with same $U_{\alpha}$ in force.

\item[(iv)] It is now clear from (iii) above that if each $W_y$ in $\mathbf{W}$
has simple eigenvalues, then $\mathbf{V}_{\alpha}$'s are
mutually disjoint. This happens for the case $d=2$ and $d=3$ but may not be so
for larger $d$ because the requirement is that eigenvalues of $W_y$ lie on the
unit circle and (counted with multiplicities) add to zero. In the last section
on examples for $d>3,$ we give concrete situations.  It is as if there is a fan
of these subsets $\mathbf{V}_{\alpha}$'s (possibly overlapping) hinged at $I$ 
and, accordingly, {\it a
fan of abelian subspaces of} $\mathcal{B}(\mathcal{H})$ (possibly overlapping),
hinged at the linear span of  $I.$

\item[(v)] A $\mathbf{W}$-MASS of size $d-1$ together with $I$ generates a
maximal abelian subalgebra of $M_d,$ in short, a MASA in $M_d.$ Theory of
orthogonal MASA is well developed by \cite{Ha}, \cite{Wei}, \cite{Ch}. In fact,
\cite{Ch} even defines an entropy $h(A/B)$ between a pair $(A,B)$ of MASA's and
proves that $A,B$ are orthogonal if and only if $h(A/B)$ takes the maximum
value, and then the value is $\log \,d.$
\end{itemize}
\end{rem}

{\bf Definitions and Discussion 2.10.}\,\,\,{(Fan representation \&
Hadamard fans).}

\begin{itemize}
 \item[(i)] In view of item 2.9 (iv) above, we call $\mathcal{V}_{\mathbf{W}}$
in Theorem 2.8 (i), the {\it fan representation of} $\mathbf{W}.$ One can
figure out the $\mathbf{W}$-MASS fan representation through a common eigenvector
system approach. It is neat when eigenvalues of  each $W_x$ in $\mathbf{W}$ are
simple and becomes quite involved when some of them are multiple. 

 \item[(ii)] The family $H_{\mathbf{W}} = \left \{H_{\alpha} : \alpha \in
\Lambda   \right \}$ facilitated as in item 2.9 (iii) above, will be called
the {\it Hadamard fan} of $\mathbf{W}.$ We
note that if $\mathbf{W}$ and $\mathbf{W}^{\prime}$ are unitary systems with
$\mathbf{W}$CUE$\mathbf{W}^{\prime}$ then their Hadamard fans are the same to
within a labelling of $\Lambda$ and permutation of rows and columns of
$H_{\alpha}.$

\item[(iii)] Let $\mathbf{U}$ be a UB and $\mathcal{W} = \{\mathbf{W} :
\mathbf{T}= (x_0, \mathbf{U}_{x_{0}}, \mathbf{W})$ is the tag of
$\mathbf{U}\,\, \mbox{at}\,\,x_0 \}$ and $\mathcal{V} = \left \{
\mathcal{V}_{\mathbf{W}} : \mathbf{W \in \mathcal{W}}\right \}.$ Then
$\mathcal{V}$ will be called the {\it fan system} of $\mathbf{U}.$ We note that
it follows from Theorem 2.7 and Theorem 2.8 that $\mathcal{V}$ is an invariant
for $\mathbf{U}$ in the sense that if $\mathbf{U}^{\prime}$ is a UB and
$\mathcal{V}^{\prime}$ is its fan system then $\mathbf{U} \sim
\mathbf{U}^{\prime}$ if and only if to within CUE  $\mathcal{V} \cap
\mathcal{V}^{\prime} \neq
\phi$ if and only if  to within CUE $\mathcal{V} = \mathcal{V}^{\prime}.$

\item[(iv)] Let $\widetilde{H}_{\mathbf{W}} = \{\widetilde{H}_{\alpha} : \alpha
\in \Lambda\}.$ We call $\widetilde{\mathbf{H}} = \{\widetilde{H}_{\mathbf{W}} :
\mathbf{T} = (x_0, U_{x_{0}}, \mathbf{W})$ is the tag of 
$\mathbf{U}\,\,\mbox{at}\,\, x_0   \}$ the {\it Hadamard fan system} of
$\mathbf{U}.$ We note that if $\mathbf{U}$ and $\mathbf{U}^{\prime}$ are UB's
with
Hadamard fan systems $\mathbf{H}$ and $\mathbf{H}^{\prime}$ respectively, and if
$\mathbf{U} \sim \mathbf{U}^{\prime}$ then $\mathbf{H} = \mathbf{H}^{\prime}.$
The converse does not hold.
\end{itemize}

\vskip0.1in
{\bf 2.11.\,\, Maximally entangled state bases:}  The question that triggered
this
paper, in fact, is the following
one in the context of maximally entangled states (MES) with phases. How to
distinguish pairwise orthogonal systems of MES using local qantum operations 
supplemented by classical communication?

If one can figure out sets of pairwise orthogonal MES, locally unitarily
connected up to global phases to the Bell basis  then the task
of distinguishing the states from the aforesaid sets is equivalent to that of
distinguishing locally the Bell states.\\

(a) We now put the question in the language used in the beginning of this
section. Let
$\{|\Psi_x \rangle : x \in X\}$ be an orthonormal basis in $\mathcal{H} \otimes
\mathcal{H}$ consisting of MES only. Do there exist unitaries $V_1, V_2 \in
\mathcal{U}(\mathcal{H}),$ a bijective function $g$ on $X$ to itself and a
function $f$ on
$X$ to $S^1$ such that $|\psi_{g(x)}\rangle = f(x) (V_1 \otimes V_2)(U_{x}
\otimes I) \Omega,$ $x \in X$ where $\{U_x : x \in X\}$ is the system $\{U_{mn}
: m, n \in Y_d\}$ as explained in Example 3.1(vii) in the next section. 

(b) In view of the item  2.1 (v)(a), there exists a system $\{V_x : x \in X \}$
of mutually orthogonal unitaries in $\mathcal{U}(\mathcal{H})$ such that
$|\psi_x \rangle = (V_x \otimes I) \Omega$ for  $x \in X$ and further, by item
2.1(vii),  for $x \in X,$  $(V_1 \otimes V_2)(U_x \otimes I) \Omega$ 
\begin{eqnarray*}
 &=& (V_1 \otimes I) (I \otimes V_2) (U_x \otimes I) \Omega\\
&=& (V_1 \otimes I) (U_x \otimes I) (I \otimes V_2)  \Omega\\
&=&  (V_1 \otimes I) (U_x \otimes I) (V_2^t \otimes I) \Omega = (V_1 U_x V_2^t
\otimes I) \Omega.
\end{eqnarray*}
Now $\widetilde{V}_2 = V_2^t$ is a unitary if $V_2$ is so. So the question
reduces to: Do there  exist unitaries $V_1, \widetilde{V}_2 \in
\mathcal{U}(\mathcal{H})$ and a function $f:X \rightarrow S^1$ such that
$V_{g(x)} = f(x) V_1 U_x \widetilde{V}_2,$ $x \in X.$

In the terminology of item 2.1(ii): Does there exist a function $\widetilde{f}$
on $X$ to $S^1$ such that $\{\widetilde{f}(x) V_x : x \in X\}$ is equivalent to
$\{U_x : x \in X\}.$ We shall utilise the results and methods given above to
answer this.

\vskip0.1in
\noindent{\bf Definition 2.12.} 
We call two unitary bases $\mathbf{U}$ and  $\mathbf{U}^{\prime}$ {\it
phase-equivalent} if there exists a function $\widetilde{f}$ on $X$ to $S^1$
such that $\{\widetilde{f}(x) U_x^{\prime} : x \in X  \}$ is equivalent
to $\{U_x : x \in X  \}.$
\vskip0.1in

\noindent{\bf Definition 2.13.}  
For subsets $\mathcal{F}$ and $\mathcal{G}$ of $\mathcal{B}(\mathcal{H})$ we
say $\mathcal{F}$ is {\it phase-collectively-unitarily equivalent} to
$\mathcal{G}$ and write $\mathcal{F}$PCUE$\mathcal{G}$  if there exists a
function $f$ on $\mathcal{F}$ to $S^{1}$ such that $f\mathcal{F}$ CUE
$\mathcal{G}$ where, $f \mathcal{F} = \{f(A) A : A \in \mathcal{F}\}.$

\vskip0.1in
\noindent{\bf Remark 2.14.} Let $\mathbf{W} = \{ W_y : y \in Y\}$ be a
unitary system and $h : Y \rightarrow
S^{1}$ be a function. Then
\begin{itemize}
 \item[(i)] $h \mathbf{W} = \{h(y) W_y : y \in W \}$ is a unitary system, 
 \item[(ii)] $\mathbf{W}$ is abelian if and only if $h\mathbf{W}$ is abelian,
 \item[(iii)] $\mathbf{V} = \{W_y : y \in Z \}$ with $Z \subset Y$ is a
$\mathbf{W}$-MASS if and only if $(h|Z) \mathbf{V}$ is a $\mathbf{W}$-MASS, and

\item[(iv)] $h \mathcal{V}_{\mathbf{W}} = \mathcal{V}_{h\mathbf{W}}.$   
\end{itemize}
 We can now have the obvious generalizations of Theorem 2.7, Theorem 2.8, item
2.9 and item 2.10 with obvious modifications of the corresponding proofs. Here
is an illustration which will be strengthened further by examples in the next
section.

\noindent{{\bf Theorem 2.15}.}
Let $\mathbf{U}, \mathbf{U}^{\prime}$ be unitary bases for $\mathcal{H}.$ Then
the following are equivalent.
\begin{itemize}
 \item[(i)] $\mathbf{U}$ is phase equivalent to  $\mathbf{U}^{\prime}.$
\item[(ii)] For some $\mathbf{U}$-associated tag $\mathbf{T}$ and some
$\mathbf{U}^{\prime}$-associated tag $\mathbf{T}^{\prime},$\\
$\mathbf{W}$PCUE$\mathbf{W}^{\prime}.$
\item[(iii)] For each $\mathbf{U}$-associated tag $\mathbf{T},$ there
is a $\mathbf{U}^{\prime}$-associated tag $\mathbf{T}^{\prime}$ such that
$\mathbf{W}$ PCUE $\mathbf{W}^{\prime}.$ 
\item[(iv)] $\mathbf{U}$ and $\mathbf{U}^{\prime}$ have the same fan systems to
within PCUE. 
\end{itemize}

\section{Examples}
The purpose of this section is to illustrate concepts and results in Section 2
and throw more light on them.
\vskip0.1in
\noindent{\bf Example 3.1. (Shift and multiply)\,\,}
 This is based on (\cite{VW}, III.A) or (\cite{Wer}, Proposition 9). Let $Y_d =
\{1,2,\ldots, d  \}.$ For $m \in Y_d,$ let $H^{(m)} = [H_{jk}^{m}]$ be a complex
$d \times d$ Hadamard matrix, i.e. $H^{(m)}$ $H^{(m)\ast} = dI_d$ and
$|H_{jk}^{m}| =1$ for all $m,j,k.$ Let $\lambda$ be a latin square, i.e.
$\lambda$ is a map on $Y_d \times Y_d$ to $Y_d$ satisfying $k \rightarrow
\lambda (k,\ell)$ and $k \rightarrow \lambda (\ell, k)$ are injective for every
$\ell.$ We shall write $e_k$ as $|k \rangle$ for $k \in Y_d.$

For $m,n \in Y_d,$ let $U_{mn}$ (or, at times, also written as $U_{m,n}$ or
$U_{(m,n)}$) be the operator which takes $|k\rangle$ to $H_{mk}^n$ $|\lambda
(n,k)\rangle,$ $k \in Y_d.$ 
\vskip0.1in
\begin{itemize}
\item[(i)] Then $\mathbf{U} = \{U_{m,n} : m, n \in Y_d \}$ is a UB. We note that
its indexing set is $X=Y_d \times Y_d.$ Further, $U_{m,n} = I$ if and
only if $\lambda(n,k)=k$ and $H_{mk}^n = 1$ for each $k \in Y_d.$ In this case
$\mathbf{W} = \mathbf{U}\backslash \{I\}$ is a unitary system.
\vskip0.1in
\item[(ii)] For $(m,n),(m^{\prime},n^{\prime}) \in X$ we say $(m,n)$
{\it commutes }
with $(m^{\prime},n^{\prime})$ and write it as $(m,n) \Delta
(m^{\prime},n^{\prime})$ if $U_{m,n}$ commutes with $U_{m^{\prime},n^{\prime}}.$
We now proceed to obtain maximal commuting subsets of $\mathbf{U}$ (to be
called $\mathbf{U}$-MASS's) or,
equivalently, of $X.$
\vskip0.1in
Let $(m,n),(m^{\prime},n^{\prime}) \in X.$ Then $U_{mn} U_{m^{\prime}
n^{\prime}} = U_{m^{\prime}n^{\prime}} U_{mn}$ if and only if for $k \in Y_d,$
$U_{mn}\left (H_{m^{\prime}k}^{n^{\prime}} |\lambda(n^{\prime},k) \rangle 
\right )  = U_{m^{\prime}n^{\prime}}\left (H_{mk}^{n} |\lambda(n,k) \rangle 
\right ) $ if and only if for $k \in Y_d,$ $ H_{m, \lambda(n^{\prime},k)}^{n} 
H_{m^{\prime}k}^{n^{\prime}}  | \lambda (n, \lambda (n^{\prime},k) )\rangle
\!\!= \!\! H_{m^{\prime},  \lambda(n,k)}^{n^{\prime}}$ $H_{mk}^n | \lambda
(n^{\prime}, \lambda(n,k)) \rangle $ if and only if 
$$\lambda (n,\lambda (n^{\prime}, k)) = \lambda (n^{\prime},
\lambda (n,k)), \,\,k \in Y_d \,\,\,\,\,\,\,\left [\lambda_{n,n^{\prime}} \right
]$$ 
and $H_{m, \lambda(n^{\prime}, k)}^{n}  H_{m^{\prime}k}^{n^{\prime}}  =
H_{m^{\prime}, \lambda(n, k)}^{n^{\prime}}  H_{mk}^{n},$ $k \in Y_d \,\,\,\,\,
\left [H_{(m,n),(m^{\prime}, n^{\prime})}\right ].$ \\
We call these conditions {\it Latin criss-cross} and {\it Hadamard
criss-cross} respectively.
\vskip0.1in

\item[(iii)] {\bf Latin squares.\,\,} A latin square $\lambda$ may be called a
quasigroup $L$ in the
sense that the binary operation `.' on $L$ given by $a.b=\lambda(a,b)$
satisfies the condition that, given $s,t \in L,$ the equations $x.s=t$ and
$s.y=t$ have unique solutions in $L;$ one may see, for instance, the book by
Smith \cite{jdhs} for more details. Keeping this in mind we introduce a few
notions for $\lambda.$\\
(a) An element $a$ of $L$ will be called a {\it left identity} if $a.b=b$ for
$b$ in
$L.$ We note that a left identity, if it exists, is unique. Similar remarks
apply to the notion and uniqueness of {\it right identity.}\\
(b)  $\lambda$ is called {\it associative} if `$\cdot$' is associative.\\
(c) Elements $a,b$ in $L$ will be said to be {\it commuting} if $a.b=b.a.$\\
(d) The {\it centre} $Z(L)$ of $L$ is $\{a \in L: a.b=b.a \,\,\mbox{for
each}\,\, b\,\, \mbox{in}\,\, L\}.$
\vskip0.1in
We shall mainly consider latin squares arising from a group $G$ (with
multiplication written as juxtaposition and identity written as $e$) or {\it
right
divisors} or {\it left divisors} in the group $G$ as follows:\\
(e) $a.b=ab,$\\
(f) $a.b=ab^{-1},$\\
(g) $a.b=a^{-1}b,$ for $a,b$ in $G.$
\vskip0.1in
Direct computations give the following. \\
(h) A right (respectively, left) divisor latin square has $e$ as right
(respectively, left) identity. Further,  any such latin square has both right
and left identity  if and only if $a^2 = e$ for each $a$ in
$G$ if and only if $L$ is the same as $G.$\\
(i) Any such $\lambda$ is associative  if and only if $G$ and $L$ coincide.\\
(j) Elements $a,b$ in any such $L$ commute if and only if  $(ab^{-1})^2 = e.$\\
(k) In particular, if the number of elements in $G,$ $|G|$ is an odd number $\ge
3$ then no two distinct
elements in any such $L$ commute.\\
We may have {\it twisted version} of (e), (f) and (g) as follows and then draw
the same conclusions as above for them.\\
(l) $a.b=ba,$\\
(m) $a.b=b^{-1} a,$\\
(n) $a.b=ba^{-1},$ for $a,b$ in $G.$\\
Let $\lambda^{-1}$ be the latin square $\mu$ defined by $\mu (a, \lambda
(a,b))=b$ for $a,b$ in $L.$\\
(o) Direct computations give that $\mu^{-1}=\lambda.$ We may say that
$(\lambda, \mu)$ is an {\it inverse-pair}. We note that latin squares listed
above may then be inverse-paired as $((e),(g)),$ $((f),(m))$ and $((l), (n)).$

\item[(iv)] {\bf Latin criss-cross.\,\,} Item (iii) above immediately gives the
following facts.
\vskip0.1in
\noindent (a) Latin criss-cross for an associative latin square reduces to
$n.n^{\prime} = n^{\prime}. n.$\\
(b) If $\lambda$ has a right identity then Latin criss-cross implies that
$n.n^{\prime}=n^{\prime}.n.$ In particular, if $\lambda$ is a right divisor
latin square and
$|L|$ is an odd number $\ge 3,$ then Latin criss-cross reduces to
$n=n^{\prime}.$\\
(c) Direct computations give that if $\lambda$ is a left divisor latin square
then  Latin criss-cross reduces  to $nn^{\prime}=n^{\prime}n.$ \\
(d) Suppose $\lambda$ is a right divisor latin square. Then Latin criss-cross
holds if  and only if  $(n^{\prime} n^{-1})^2 =e,$
$n^{\prime}n=nn^{\prime}$ and $n^{\prime}n^{-1} \in Z(G).$ For the sake of
illustration we give details. We first note
that Latin criss-cross holds if and only if $n^{\prime} kn^{-1} =
nkn^{\prime -1}$
for all $k.$ Taking $k=e, n, n^2$ this implies $n^{\prime} n^{-1}=nn^{\prime
-1},$
$n^{\prime} = n^2 n^{\prime -1}$ and $n^{\prime}n=n^3 n^{\prime -1}$
which, in turn,
implies $(n^{\prime} n^{-1})^2=e,$ $n^{\prime 2} = n^2$ and $n^{\prime}n = n^3
n^{\prime -1}$ and thus, $(n^{\prime} n^{-1})^2=e,$ $n^{\prime 2} = n^2,$
$n^{\prime}n=nn^{\prime}.$ This is equivalent to $n^{\prime}n = nn^{\prime},$ 
$n^2 = n^{\prime 2}$ as also to $(n^{\prime} n^{-1})^2=e,$ $n^{\prime}n =
nn^{\prime}.$ Therefore, Latin criss-cross is equivalent to $n^{\prime}n =
nn^{\prime},$ $n^{\prime} n^{-1} = nn^{\prime -1},$ $n^{\prime} kn^{-1} =
nkn^{\prime -1}$
for all $k,$ i.e., $n^{\prime}n =nn^{\prime},$ $n^{\prime} n^{-1} = nn^{\prime
-1} =
n^{\prime -1}n,$ $n^{\prime} n^{-1}(nk) = (nk) n^{\prime -1} n$ for all $k,$
\,\,
i.e., $n^{\prime} n=nn^{\prime},$ $n^{\prime} n^{-1}=nn^{\prime -1},$
$(n^{\prime}
n^{-1})\ell = \ell (n^{\prime} n^{-1})$ for all $\ell$
i.e., $n^{\prime}n = nn^{\prime},$ $(n^{\prime} n^{-1})^2=e,$ $n^{\prime}
n^{-1} \in Z(G).$\\
(e) If $\lambda$ is a right divisor latin squre and $Z(G)$ consists of the
identity then Latin criss-cross reduces to $n=n^{\prime}.$\\
(f) Suppose $G$ is an abelian group with identity written as $0$ and $|G|>2.$
Then (c)
above gives that Latin criss-cross is satisfied automatically for the left
subtraction latin square. Also (d) above gives that Latin criss-cross for the
right subtraction latin square is satisfied if and only if $n+n = n^{\prime}
+n^{\prime}.$ As already noted in (b) above, it is possible for $n \neq
n^{\prime}$ only if $|G|$ is even and in that case, for $0 \neq g \in G,$
order of $g$ even, say, $2r,$  $n=0,$  $n^{\prime} = rg
(=\underset{r \,\,\mbox{times}}{\underbrace{g+\cdots+g}} )$ satisfy the
requirement. We now assume that $|G|$ is even.  We can divide $L$ into mutually
disjoint equivalence classes $L_h$, indexed by the set $S=\{h=g+g\,\,:\,\, g \in
G\}$, given by $L_h=\{n \in G: n+n=h\}.$ We note that $L_h = \cup \{ L_0 + g :
g+g=h\}$ and Latin criss-cross is satisfied if and only if $n,n^{\prime}$ both
belong to some $L_h.$ In case $\exp G=2,$ we have $L=G,$ $S=\{0\},$ $L_0=L$ and,
therefore, Latin criss-cross is automatically satisfied. On the other hand, if
$\exp G > 2,$ then $\{0\} \subsetneqq S, \{0\} \subsetneqq L_0 \subsetneqq L$
and $|L_0| \leq |L_h|$ for $h \in S.$ 

\vskip0.1in
\item[(v)] {\bf Hadamard criss-cross.}\\
(a) We first consider the case $n=n^{\prime}.$  Hadamard criss-cross becomes
$$H_{m,n.k}^{n} H_{m^{\prime},k}^{n} =  H_{m^{\prime}, n.k}^{n}
H_{m,k}^{n}\,\,\mbox{for each}\,\,k, $$
i.e., $\frac{H_{m^{\prime},n.k}^{n}}{H_{m, n.k}^{n}} =
\frac{H_{m^{\prime},k}^{n}}{H_{m,k}^{n}}$ for each $k.$
\vskip0.1in
(b) Suppose $L$ is a left divisior latin square. Then for $n=n^{\prime}=e,$
Hadamard criss-cross becomes an identity and
therefore,
$m, m^{\prime}$ can be arbitrary.  This gives rise to a full-size
$\mathbf{U}$-MASS.
\vskip0.1in
(c) We now consider the case when $n=n^{\prime}$ is a must. For  $j, k \in
\mathbb{Z}_d,$ let $H_{j,k}^{n} = (\eta_d)^{jk},$ where $\eta_d = \exp (2 \pi
i/d).$ Hadamard criss-cross is equivalent to
$(m-m^{\prime})(n.k-k)=0~ (\mbox{mod}\,d)$ for all $k.$ If $(n.k - k)$ is co-prime
to $d$ for some $k,$ then Hadamard criss-cross
holds if and only if $m=m^{\prime}.$ This means $(m,n)$ belongs to the unique
$\mathbf{U}$-MASS $\{(m,n)\}$ containing it. In particular, it is so if $n$
is not the left identity for $\lambda$ and $d$ is a prime. We record important
consequences. 
\vskip0.1in
(d) {\bf Singleton $\mathbf{U}$-MASS's example.\,\,} We consider the
right subtraction latin square coming from
$\mathbb{Z}_3$ and take, for  $n=0,1,$ or $2,H^n=\left [\begin{array}{ccc} 1 & 1
&
1
\\ 1 & \omega & \omega^2 \\ 1 & \omega^2 &  \omega \end{array} \right ].$
\vskip0.1in
By (iv)(b), Latin criss-cross is satisfied if and only if $n^{\prime}=n.$ By (v)
above,  Hadamard criss-cross is satisfied if and only if $\frac{H_{m^{\prime},
n-k}^{n}}{H_{m, n-k}^{n}} = \frac{H_{m^{\prime},k}^{n}}{H_{m,k}^{n}}$ for each
$k.$  For $n=1$ or 2, taking $k=0$ this forces
$\frac{H_{m^{\prime},n}^{n}}{H_{m,n}^{n}}=1$ i.e. $H_{m^{\prime},n}^{n} =
H_{m,n}^{n},$ which, in turn, forces $m^{\prime}=m.$ So $\mathbf{U}$-MASS's are
all of size $1.$
\vskip0.1in
(e) We consider the non-abelian group $G=S_3$ of permutations on $\{a, b, c\}.$
We label the elements of $G$ in any manner by $\{0,1,2,3,4,5\}$ but with $0=e,$
$1= \,\,\mbox{the cycle}\,\, (ab).$ Set $\mathbf{W}= \mathbf{U} \backslash 
\{(0,0)\}.$  Then  $n=1=k$ satisfy the requirement that $(n.k-k)$ is co-prime to $d.$ 
So for $0 \leq m \leq 5,$ $\{(m,1)\}$ is a $\mathbf{W}$-MASS. We 
can label the
cycle (ac) as 3 and cycle  (bc) as 5, then similar arguments give that for $0
\leq m \leq 5,$ $\{(m,3)\}$ and $\{(m,5)\}$ are $\mathbf{W}$-MASS's.  Thus, we 
have 18 
$\mathbf{W}$-MASS's of size one each, say, $F_{1}, F_{2}, \ldots,  F_{18}.$    
Labelling 
$(ab)(bc) = (abc)$ as 2 and $(ab)(ac)=(acb)$ as 4,  we 
have four full-size 
$\mathbf{W}$-MASS's viz., $F_0=\{(m,0): 1 \leq m \leq 5\},$  $F_{19} = \{(0,2), 
(3,2), (0,4), (3,4), (3,0) \},$ $F_{20} = \{(1,2), (4,2), (2,4), (5,4), 
(3,0)\}$ and $F_{21} = \left \{(2,2), (5,2), (1,4),(4,4),\right .$ \\$\left .  
(3,0)\right \}.$ The element $(3,0)$ is present in all these four 
$\mathbf{W}$-MASS's. All other elements belong to a unique $\mathbf{W}$-MASS. 
Figure 1 gives an idea.

\begin{center}
\begin{figure}[ht]
\includegraphics[width=6in]{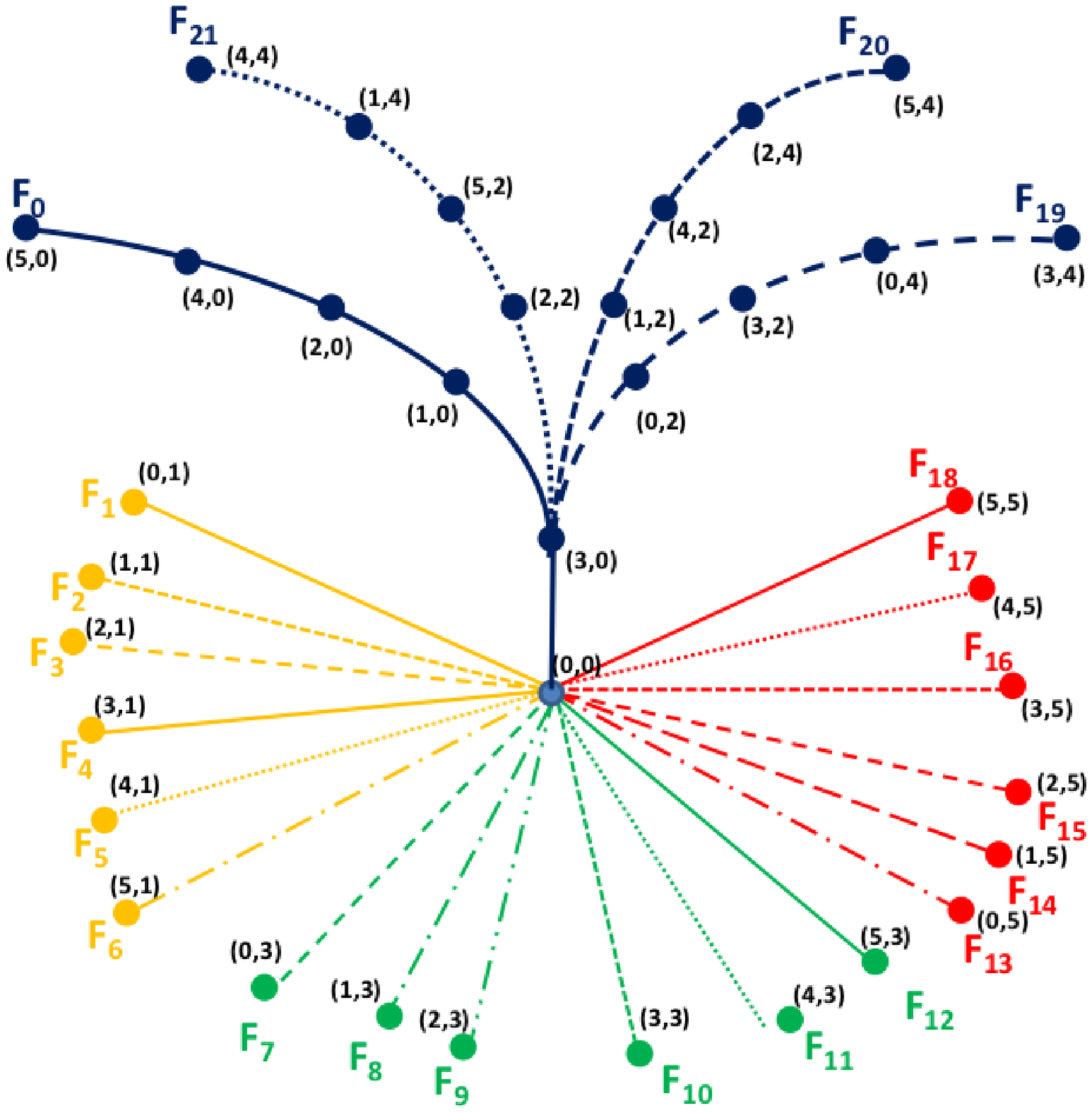}  
\caption{}
\end{figure}
\end{center}

\vskip0.1in
\item[(vi)]  {\bf Tags and Twills \,\,} Contents of Section 2 tell us that it
is $\mathbf{W}$-MASS's for tags of $\mathbf{U}$ that really help us. And
$\mathbf{U}$-MASS's help directly if (a scalar multiple of)  $I$ is in 
$\mathbf{U}$ simply because then apart from (the scalar multiple of) $I$
occurring in all $\mathbf{U}$-MASS's, $\mathbf{W}$-MASS's and
$\mathbf{U}$-MASS's
are same. Let $\mathbf{T} = (x_0, U_{x_{0}}, \mathbf{W})$ be a tag of
$\mathbf{U}.$  As noted in Remark 2.2, for $x,y \in X \backslash \{x_0\},$
$W_x W_y = W_y W_x$ if and only if Twill $\mathbf{T}(x, x_0, y)$ viz., $U_x
U_{x_{0}}^{\ast}
U_y = U_y U_{x_{0}}^{\ast} U_x$ is satisfied. We now figure these out for some
of the
cases considered above.\\
(a) Let $x_0 = (m_0, n_0).$ Then for $k \in Y_d,$ 
$$U_{x_{0}}^{\ast}
|k \rangle = U_{x_{0}}^{-1} |k \rangle = \overline{H_{m_{0},
\mu (n_{0}, k)}} | \mu (n_0, k) \rangle.$$ 
So the Twill is equivalent to $\left [\lambda (n, n_{0}, n^{\prime}) \right ]$ and $\left [H (m,n), (m_0, n_0), (m^{\prime}, n^{\prime})
\right ]$, given by  
$$\lambda (n, \mu (n_0, \lambda (n^{\prime}, k)) = \lambda (n^{\prime}, \mu
(n_{0}, \lambda (n,k)), k \in Y_d,$$ 
and
$$H_{m, \mu (n_{0}, \lambda (n^{\prime}, k))}^{n} \overline{H_{m_{0},
\mu (n_{0}, \lambda (n^{\prime}, k))}^{n_{0}}} H_{m^{\prime},k}^{n^{\prime}}
$$
$$= H_{m^{\prime}, \mu (n_{0}, \lambda (n,k))}^{n^{\prime}}
\overline{H_{m_{0}, \mu (n_{0}, \lambda(n,k))}^{n_{0}}} H_{m,k}^{n},\,\, k \in Y_d,$$
respectively.
\vskip0.1in
We call them  {\it Latin twill} and {\it Hadamard twill} respectively.
\vskip0.1in
We may re-write Hadamard twill in another form as
\begin{eqnarray*}
\lefteqn{H_{m, \mu (n_{0}, \lambda (n^{\prime}, k))}^{n} H_{m_{0},
\mu (n_{0}, \lambda (n, k))}^{n_{0}} H_{m^{\prime},k}^{n^{\prime}}
}\\
&=& H_{m^{\prime}, \mu (n_{0}, \lambda (n,k))}^{n^{\prime}}
H_{m_{0}, \mu (n_{0}, \lambda(n^{\prime},k))}^{n_{0}} H_{m,k}^{n}, k \in Y_d.
\end{eqnarray*}
\vskip0.1in
(b) For latin squares coming from group $G$ as in (iii) (e) above. Latin twill
reduces to 
$$nn_{0}^{-1}n^{\prime} = n^{\prime} n_{0}^{-1} n. $$
For the inverse latin square arising as in (iii)(g) above, it is $n^{-1} n_{0}
n^{\prime -1} = n^{\prime -1} n_{0} n^{-1},$ which on taking inverses, becomes
$$n^{\prime} n_{0}^{-1} n = nn_{0}^{-1} n^{\prime}.$$
Thus the two Latin twills are the same. Interestingly, the Latin twill remains
the same for latin squares arising as in (f), (l), (m), (n) as well. We note an
equivalent useful form of the Latin twill: $(n^{-1}_0 n)(n_0^{-1}
n^{\prime})=(n_0^{-1} n^{\prime})(n_0^{-1}n).$ This Latin twill is satisfied
automatically for abelian groups  $G.$\\
(c) Hadamard twills can be written down which will be quite complicated for
different cases.
\vskip0.1in
 We
consider the case in (v)(d) coming from $\mathbb{Z}_3.$ Then $\lambda \equiv
\mu$ and $\lambda (n,k) = n-k =\mu (n,k)$ for $n,k=0,1,2.$ So Hadamard twill
becomes: for $k \in \{0,1,2\},$
\begin{eqnarray*}
\lefteqn{H_{m, n_{0}-n^{\prime} + k}\,\, H_{m_{0}, n_{0} - n + k}\,\,
H_{m^{\prime},k} } \\
&=& H_{m^{\prime}, n_{0} - n+ k} \,\,H_{m_{0}, n_{0}-n^{\prime}+k}\,\,
H_{m,k},
\end{eqnarray*}

i.e., for $k \in \{0,1,2\},$ 
\begin{eqnarray*}
\lefteqn{ m(n_0-n^{\prime}+k) + m_0 (n_0-n+k)+ m^{\prime} k    }\\
&=& m^{\prime} (n_0-n+k) + m_0 (n_0-n^{\prime}+k) + mk \quad (\mbox{mod}\,3)
\end{eqnarray*}
i.e., for $k \in \{0,1,2\},$ 
\begin{eqnarray*}
\lefteqn{m (n_0 - n^{\prime}) + m_0 (n_0-n) + (m+m_0+m^{\prime}) k    }\\
&=& m^{\prime}(n_0-n)  + m_0 (n_0 - n^{\prime})+ (m+m_0+m^{\prime}) k
\end{eqnarray*}
i.e.,
$$(m-m_0) (n^{\prime}-n_0)-(m^{\prime}- m_0)(n-n_0)= 0  \quad (\mbox{mod}\,3).
$$
This  gives us exactly  $4$ $\mathbf{W}$-MASS's $\{(m_0, k) : k \neq n_0\},$
$\{(j, n_0): j \neq m_0  \},$ $\{(m_0 + 1, n_0+1), (m_0+2, n_0+2)\}$ and
$\{(m_0+1, n_0+2), (m_0+2, n_0+1) \}.$ They are all full-size and mutually
disjoint.
\vskip0.1in
\item[(vii)] {\bf Cyclic  group case.\,\,\,} We now consider the case when
$L$ is the cyclic group $\mathbb{Z}_d = \{0,1,\ldots,d-1\}$ with addition mod
$d$  and for $n  \in
\mathbb{Z}_d,$ $H^n$ is as in (v)(c) above 
(\cite{Iv}, \cite{WF}, \cite{BBRV}, \cite{Ha}, \cite{PDB} and \cite{Sz}). It is
as if we have now taken $X = \mathbb{Z}_d \times \mathbb{Z}_d.$ Then $U_{00}=I.$ 
This is the so-called Bell basis.

The commuting condition becomes
$$mn^{\prime} - m^{\prime}n \equiv 0 (\mbox{mod} \, d).$$
One may see more details in the papers referred to, particularly for $d=p^r,$
where $p$ is a prime and $d=6 \,\,\mbox{or}\,\,10.$
\end{itemize}

We consider $\mathbb{Z}_d$ as a commutative ring with addition and
multiplication moduluo $d$ and set $Y = X \backslash \{(0,0)\}.$  Then
$\mathbf{W}= \{U_x : x \in Y \}$ is a unitary system and $((0,0), I,
\mathbf{W})$ is the $\mathbf{U}$-associated tag at $(0,0).$ 

We proceed to determine various $\mathbf{W}$-MASS's.\\
(a) We note that $(0,0) \Delta x$ for each $x$ in $X.$ We further note
some obvious commuting pairs,viz., $(j, \ell j) \Delta (j^{\prime}, \ell
j^{\prime}),$  $(\ell j, j) \Delta (\ell j^{\prime}, j^{\prime}),$  for
$j,j^{\prime},$ $\ell \in \mathbb{Z}_d.$ We next set 
\begin{eqnarray*}
Y_{\ell} &=& \{(j, \ell j) : j \in \mathbb{Z}_d, j \neq 0 \} \,\,\mbox{for} \,\,
0 \leq \ell \leq d-1,\\
\widetilde{Y}_{\ell} &=& \{(\ell j, j) : j \in \mathbb{Z}_d, j \neq 0 \}
\,\,\mbox{for} \,\, 0 \leq \ell \leq d-1.
\end{eqnarray*}
We note that $\widetilde{Y}_0,$ $Y_{\ell},$ $0 \leq \ell \leq d-1$ are all
distinct. Next $Y_1 = \widetilde{Y}_1$ is the diagonal of $Y,$ whereas $Y_{d-1}
= \widetilde{Y}_{d-1}$ is the antidiagonal of $Y.$   We further note that if 
$k, \ell$ in $\mathbb{Z}_d
\backslash \{0\}$ satisfy $k \ell = 1,$ then 
$\widetilde{Y}_k = Y_{\ell}.$ In particular, if $d$ is a prime, then we may
take, for
$k \in \mathbb{Z}_d \backslash \{0\},$ $\ell = k^{-1};$ in this case
$\widetilde{Y}_0, Y_{\ell},$ $0 \leq \ell \leq d-1$ are mutually disjoint as
well. On the other hand, for a composite $d,$ certain distinct
$Y_{j}$'s do have overlaps.

We denote $\widetilde{Y}_0$ by $Y_{-1}$ and $\widetilde{Y}_{\ell}$ by 
$Y_{-\ell}$ for $2 \leq \ell \leq d-2.$

\vskip0.1in
Each of the sets $Y_{j}$ is a maximal commuting set in $Y$ and accordingly
$\{U_x
: x \in Y_j  \}$ is a $\mathbf{W}$-MASS. Their number is 3 for $d=2$ and
$2(d-1)$ for $d>2.$ Each of the sets $Y_j$ has cardinality $d-1.$ So for $d>3,$
they do overlap for some combinations. Sets  $Y_{-1},$ $Y_0,$ and $Y_{1}$ are
pairwise disjoint. Some combinations of sets may have non-empty pairwise
intersection. For instance, for $2<d = 2 r,$ $(r,r) \in Y_1
\cap Y_{d-1}$ and  $(r,0) \in Y_0 \cap Y_2,$ $(2,0) \in Y_0 \cap Y_r,$ $(2,2)
\in Y_1 \cap Y_{r+1}.$ However, for odd $d,$ $Y_{-1},
Y_{0},Y_1, Y_{d-1}$ are all pairwise disjoint, whereas, for even $d > 2,$ they
are
all pairwise disjoint except for $Y_1,$ $Y_{d-1}$ containing one element, viz.,
($\frac{1}{2}d,$ $\frac{1}{2}d$) in common. 
\vskip0.1in
Finally, for $j$ co-prime with $d,$ we have unique $\mathbf{W}$-MASS's
$Y_{-1}, Y_{0}, Y_{1}, Y_{d-1}$ for $(0,j),(j,0),(j,j)$ and $(j, d-j)$
respectively.

\vskip0.1in
(b) {\bf Move together and stand-alone technique.\,\,} We note that 
$$(m,n) \Delta (m^{\prime},n^{\prime}) \,\,\mbox{iff}\,\, (d-m, d-n) \Delta
(m^{\prime}, n^{\prime}).$$

Thus $(m,n),(d-m, d-n)$ move together in any $\mathbf{W}$-MASS. Now
$(m,n) =
(d-m, d-n)$ if and only if $d$ is even, say, $2r,$ and $m,n \in \{0,r\};$   and
then each of
$(r,0),(0,r),(r,r)$ can be termed as a ``stand alone''. We continue with the
case $d=2r.$

Now $(r,0)\Delta(m^{\prime}, n^{\prime})$ if and only if $rn^{\prime} \equiv 0
(\mbox{mod}\,d)$ if and only if $n^{\prime}$ is even.

Similarly, $(0,r) \Delta(m^{\prime}, n^{\prime})$ if and only if  $m^{\prime}$
is even. 

Next $(r,r) \Delta (m^{\prime}, n^{\prime})$  if and only if
$m^{\prime}-n^{\prime}$ is even. In particular, $(r,0) \Delta (0,r)$ if and only
if $r$ is even if and only if $(r,0) \Delta (r,r)$ if and only if $(0,r)
\Delta
(r,r).$ Finally, $(2,0) \Delta (m^{\prime},n^{\prime})$ if and only if
$n^{\prime} = 0$ or $r,$ $(0,2) \Delta (m^{\prime},n^{\prime})$ if and only if
$m^{\prime} = 0$ or $r,$ and $(2,2) \Delta (m^{\prime},n^{\prime})$ if and only
if $n^{\prime}=m^{\prime}$ or $m^{\prime}+r$ $(\mbox{mod}\,\,d).$

We utilize these observations to compute $\mathbf{W}$-MASS's for $d=4$
and $6.$

\vskip0.1in
(c) $\mathbf{d=4}.$ Move togethers are
\begin{eqnarray*}
 \{(1,0), (3,0)\}, && \{(0,1), (0,3) \} \\
 \{(1,1), (3,3)\} &&  \\
 \{(2,1), (2,3)\}, && \{(1,2), (3,2) \} \\
 \{(3,1),(1,3)\}. && 
\end{eqnarray*}

These are all disjoint. So we can try to extend them to $\mathbf{W}$-MASS's by
adjoining one
of $(2,0), (0,2)$ or $(2,2).$ We get
\begin{eqnarray*}
 \{(1,0), (2,0), (3,0)\}, && \{(0,1), (0.2), (0,3) \} \\
 \{(1,1), (2,2), (3,3)\}, &&  \\
 \{(2,1), (2,3), (0,2)\}, && \{(1,2), (3,2), (2,0) \} \\
 \{(3,1),(1,3), (2,2)\} && 
\end{eqnarray*}
and finally, $\{(2,0),(0,2),(2,2)\}.$

This gives us 7 $\mathbf{W}$-MASS's with overlaps coming  from  $(2,0), (0,2)$
or $(2,2),$ each one thrice. Figure 2 illustrates the situation.

\begin{center}
\begin{figure}[ht]
\includegraphics[width=6.5in]{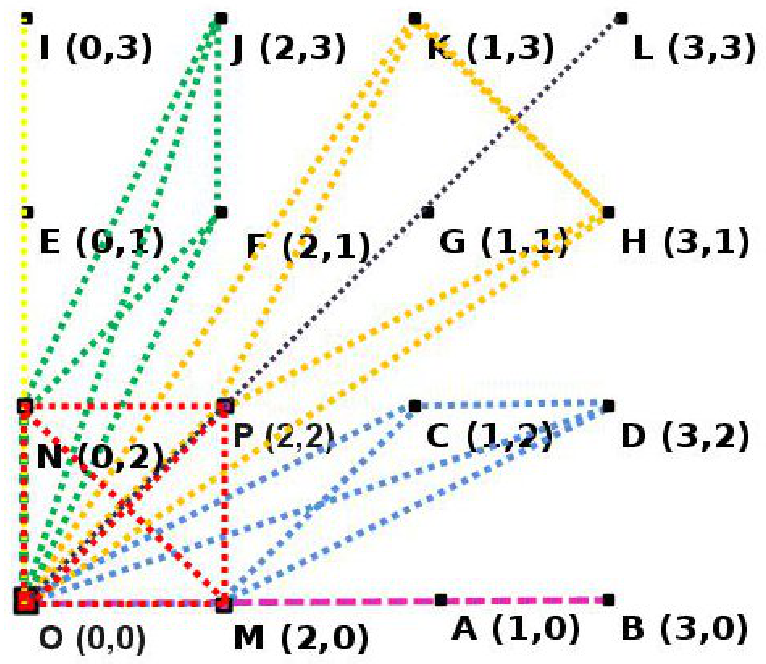}  
\caption{}
\end{figure}
\end{center}

(d) For $d,$ a product  of two co-prime  numbers $f,$ $g \ge 2,$ let $Z_{f,g} =
\{(kf, \ell g) : 0 \leq k < g, 0 \leq \ell < f, (k, \ell) \neq (0,0)\}.$ Then
$Z_{f,g}$ is a commuting set, has cardinality $d-1$ and is, therefore, a maximal
commuting set in $Y.$ Sets $Y_{\ell},$   $-(d-2) \leq \ell \leq d-1,$ may not
be enough even to cover $Y$. We can check that this is the situation for $d=6,$
and $Z_{2,3}$ and  $Z_{3,2}$ render help for this purpose.
\vskip0.1in
(e) $\mathbf{d=6}.\,\,$ The technique outlined in (b) above gives us that the
twelve $\mathbf{W}$-MASS's given by (a) and (d) above constitute the fan
$\mathcal{V}_{\mathbf{W}}.$ Here is the list:

\begin{eqnarray*}
Y_{-4} & = & \left \{(4 j,j) : j =1,2,3,4,5 \right \}\\
&& = \left \{(4,1), (2,5); (2,2), (4,4); (0,3) \right \}, \\
Y_{-3} &=& \left \{(3 j, j) : j = 1,2,3,4,5 \right \} \\
&& = \left \{(3,1), (3,5); (0,2), (0,4); (3,3) \right \},\\
Y_{-2} &=& \left \{(2j,j) : j=1,2,3,4,5 \right \} \\
&& =  \left \{(2,1), (4,5); (4,2), (2,4); (0,3) \right \}, \\
Y_{-1} &=& \left \{(0,1), (0,2), (0,3), (0,4), (0,5)\right \} 
\end{eqnarray*}
 
\begin{eqnarray*}
&& =  \left \{(0,1), (0,5); (0,2), (0,4); (0,3) \right \}, \\
Y_{0} &=& \left \{(1,0), (2,0), (3,0), (4,0), (5,0) \right \} \\
&&= \left \{(1,0), (5,0); (2,0), (4,0); (3,0) \right \}, \\
Y_{1} &=& \left \{(1,1), (2,2), (3,3), (4,4), (5,5) \right \} \\
&&= \left \{(1,1), (5,5); (2,2), (4,4); (3,3) \right \}, \\
Y_{2} &=& \left \{(1,2), (5,4); (2,4), (4,2); (3,0) \right \}, \\
Y_{3} &=& \left \{(1,3), (5,3); (2,0), (4,0); (3,3) \right \}, \\
Y_{4} &=& \left \{(1,4), (5,2); (2,2), (4,4); (3,0) \right \}, \\
Y_{5} &=& \left \{(1,5), (5,1); (2,4), (4,2); (3,3) \right \},\\
%
Z_{2,3} &=& \left \{(0,3), (2,0), (2,3), (4,0), (4,3) \right \} \\
&& = \left \{(2,0), (4,0); (2,3), (4,3); (0,3) \right \}, \\
Z_{3,2} &=& \left \{(0,2), (0,4), (3,0), (3,2), (3,4) \right \}\\ 
&& = \left \{(0,2), (0,4); (3,2), (3,4); (3,0) \right \}. 
\end{eqnarray*}

$\mathbf{W}$-MASS's are all of full-size. We have already seen in (a) above that
they do overlap and the list above makes it all clear. Figure 3 gives an idea,
where we have counted two points in a
move-together as one as per our convenience for the figure.

\begin{center}
\begin{figure}[ht]
\includegraphics[width=4in]{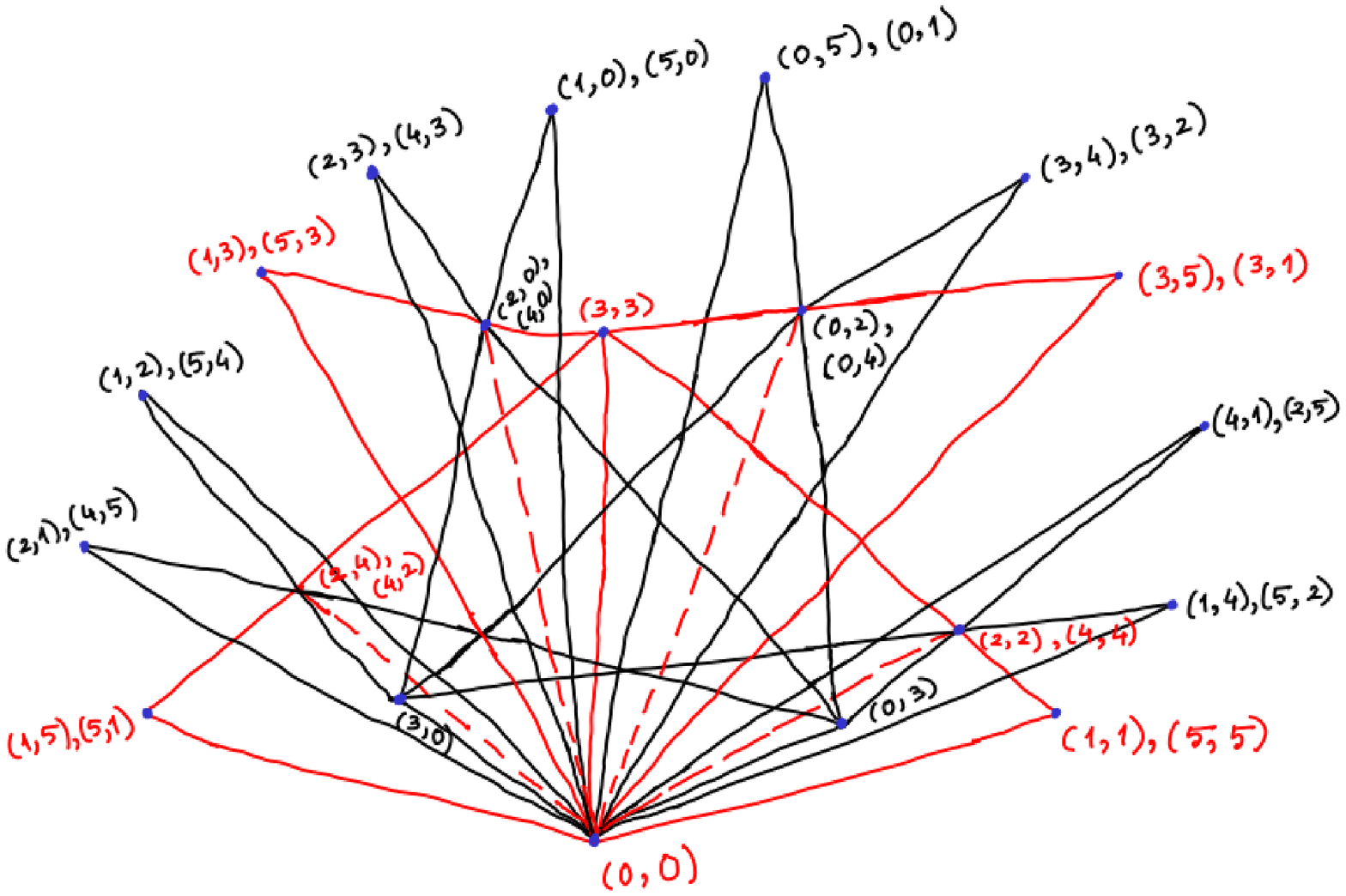}
\caption{}
\end{figure}
\end{center}

(f) Computation details in item (vi)(c) above make it clear that
$\mathbf{W}$-MASS's for tag at $(m_0, n_0)$ for the right subtraction latin
square arising from $\mathbb{Z}_d$ and all Hadamard matrices $H^n$ same as in
(v)(c) are
governed by the commuting rule
$(m-m_0)(n^{\prime}-n_0)-(m^{\prime}-m_0)(n-n_0)=0~ (\mbox{mod}\,d).$  By
Theorem 2.8 unitary basis for that and the one here are equivalent.
\vskip0.1in
(g) The descriptions of $\mathbf{W}$-MASS's for examples in (v)(e) and 
(vii)(e) above (together with Item 3.4 (iv) (b) below) make it clear that 
the two unitary bases given in these two examples for $d=6$ are not equivalent.
\vskip0.1in
\noindent{\bf Example 3.2(Pauli matrices technique). \,\,} This works for $d,$ a
power of $2$ and is based on [7],[8],[9],[10], Sych and Leuchs \cite{dsgl} and
\cite{psmw}. It
could come also under ``Shift and Multiply''-type by considering the group
$\mathbb{Z}_2 \times \mathbb{Z}_2$ and real Hadamard matrix. But it is
interesting to display the use of Pauli matrices as done in papers
cited above, particularly, \cite{dsgl} and
\cite{psmw}. \\
(i) It follows immediately that there are 15 $\mathbf{W}$-MASS's 

\vskip0.1in
of the type 
\begin{eqnarray*}
\mathcal{C}_X &=& \left \{X \otimes I, I \otimes X, X \otimes X   \right \}, \\
\mathcal{D}_X &=& \left \{X \otimes X, Y \otimes Z, Z \otimes Y   \right \}, \\
\mathcal{C}_{YZ} &=& \left \{Y \otimes I, I \otimes Z, Y \otimes Z \right\}, \\
\mbox{and,}\quad\mathcal{D}_{YZ} &=& \left \{I \otimes Y, Z \otimes I, Z \otimes
Y \right\}; \\
\mbox{together with}\quad\mathcal{E} &=& \left \{X \otimes X, Y \otimes Y,
Z \otimes Z \right\}, \\
\mathcal{F} &=& \left \{X \otimes Y, Y \otimes Z, Z \otimes X \right\}, \\
\mbox{and,}\quad\mathcal{G} &=& \left \{Y \otimes X, Z \otimes Y, X \otimes
Z \right\}.
\end{eqnarray*}
Each unitary in $\mathbf{W}$ occurs in exactly three of them. Figure 4 gives
an idea.\\
(ii) To within phases of $1$ and $-1$ and relabelling, all tags have the same
underlying unitary system $\mathbf{W}.$\\
(iii) To within PCUE, the fan system comes from (i).\\
(iv) By Theorem 2.8 (together with Item 3.4 (iv) (b) below), the unitary basis here is not equivalent to the one in 
Example 3.1 (vii)(c).

\begin{center}
\begin{figure}[ht]
\includegraphics[width=6in]{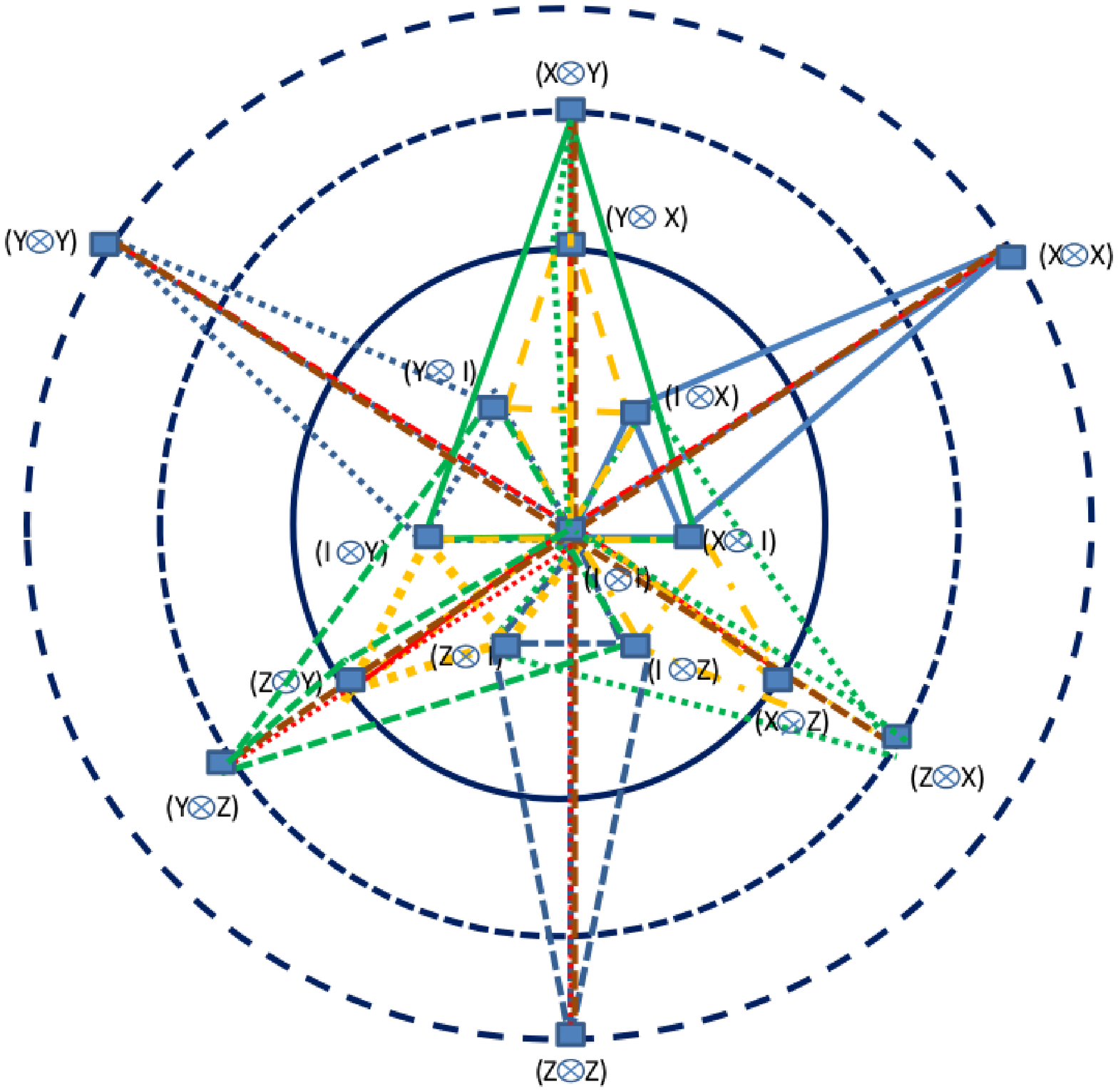}  
\caption{}
\end{figure}
\end{center}

\noindent{\bf Remark 3.3} \,\, The question of phase equivalence in the examples
above will not present significantly new points because it amounts to
multiplying different rows of the Hadamard matrix by different numbers of
modulus one.
If the latin square has a right identity $k_0,$ then we can normalize this
situation
by keeping the $k_0$-column in each Hadamard matrix consisting of one's
alone. In the particular case when $\lambda$ comes from a group, we may choose
the
identity to be the first element and thus insist on the first row and the first
column of each Hadamard matrix to consist of one's alone.
\vskip0.1in
\noindent{\bf Example 3.4}\\
For phase-equivalence the best set up is perhaps of nice unitary error bases
defined by Knill \cite{EK}.
\vskip0.1in
(i)As in (\cite{EK}, \S 2) a nice unitary error basis on a Hilbert space
$\mathcal{H}$ of  dimension $d$
is defined as a set $\mathcal{E} = \{E_g\}_{g \in G}$ where $E_g$ is unitary
on $\mathcal{H},$ $G$ is a group of order $d^2,$ $e$ its identity, ${\rm tr}E_g
= d \delta_{g,e}$ and $E_gE_h = \omega_{g,h} E_{g  h}.$  By renormalizing the
operators of the error basis, it can be assumed that $\det E_g=1,$ in which case
$\omega_{g,h}$ is a $d$-th root of unity. Error bases with this property are
called {\it very nice.} Such error bases generate a finite group of unitary
operators $\overline{\mathcal{E}}$  whose centre consists of scalar multiples of
the identity. An {\it error group} is a finite group of unitary operators
generated by a nice unitary error basis and certain multiples of the identity.
The group $\mathbf{H}$ is an {\it abstract error group} if it is isomorphic to
an error group.
\vskip0,1in
(ii) We quote Knill's Theorem without proof. 
\vskip0.1in
\noindent{\bf Theorem } (\cite{EK}, Theorem 2.1).\,\, {\it The finite group
$\mathbf{H}$ is an abstract error group if and only if  $\mathbf{H}$ has an
irredcucible character supported on the centre and the kernel of the associated
irredcucible representation is trivial.}
\vskip0.1in
(iii) These concepts have been intensively and extensively studied by
researchers and also very efficiently utilised by some of them for constructing
interesting examples of error-detecting (correcting) quantum codes. For this
purpose, the rich theory of group actions, Weyl operators, Weyl commutation
relations, multipliers, cocycles, bicharacters, imprimitivity systems has been
found to be of great importance by them, particularly by K. R. Parthasarathy
(who himself has contributed significantly to the theory for several decades, in
fact). For a good account we may refer to his recent book \cite{Pa} and
references like \cite{dg}, \cite{vakrp} and \cite{vapkkrp} therein. 
\vskip0.1in
(iv) The underlying projective representation in (i) viz., $g \rightarrow E_g$
leads to some very useful facts. \\
(a) For $g \in G$, $\omega_{g,g^{-1}} = \omega_{g^{-1}, g}$ and
$E^{\ast}_{g} =
E^{-1}_{g} = \overline{{\omega}_{g, g^{-1}}}
E_{g^{-1}}.$ \\
(b) For all tags $\mathbf{T},$ the underlying unitary system
$\mathbf{W}$ is the same up to relabelling and phases.
This permits us to consider
the fan system the same as $\mathcal{V}_{\mathbf{W}}$ for any $\mathbf{W},$ so
as
to say. In particular, we may drop $\mathbf{W}$ from $\mathbf{W}$-MASS. In
fact, it is enough to consider $\mathbf{W} = \{E_g : g \in G, g \neq e\}.$
Further, figures above display the respective fan systems as well.\\
(c)
$E_g, E_{g^{-1}}$ move
together in any $\mathbf{W}$-MASS. We now proceed to strengthen this
observation. 

\vskip0.1in
(v) Let $G$ be a group and $e$ its identity and $M$ a maximal
commutative subset of $G.$ Then $M$ is a subgroup of $G.$ To see this
well-known basic fact in group theory, we first
note that $M$ can not be empty simply because for a in $G,$  $\phi \subset
\{a\}$ which is
commutative. Now let $g_1, g_2 \in M.$ Then for $h \in M,$ $(g_1 g_2^{-1})h =
g_1 (g_2^{-1}h) = g_1 (h g_2^{-1}) = (g_1 h) g_2^{-1} = h(g_1g_2^{-1}).$ So by
maximality of $M,$ we have $g_1 g_2^{-1} \in M.$ This gives that $M$ is a
subgroup of $G.$ We may say that $M$ {\it is a maximal commutative subset of $G$
if and only if it is a maximal abelian subgroup of $G.$}
\vskip0.1in
(vi) Let $\mathbf{H}$ be a nice error group  arising from a very nice
error basis
as in (i) above.
\vskip0.1in
We write $\omega_{g,h}$ by $\omega(g,h)$ and also $ \mathbf{U}  =\{U_g : g \in
G\}$  instead of $\mathcal{E}$ for notaional convenience. Let
$\mathbf{T}_{\omega}$ be the subgroup of $S^1$ generated by the range of
$\omega.$ Then $\mathbf{H} =
\{(g,\alpha) : g \in G, \alpha \in \mathbf{T}_{\omega}\}$ and,  for $(g,
\alpha),
(h, \beta) \in \mathbf{H},$ $(g, \alpha) (h, \beta)= (gh, \omega (g,h)\alpha
\beta).$ Because $\mbox{tr}\,(U_g^{-1} U_h) = \delta_{g,h} d$ for $g,h$ in $G,$
whenever $U_g = \lambda U_h$ for some $g,h$ in $G$ and scalar $\lambda,$ we
must have $g=h$ and $\lambda = 1.$ So for $g, h \in G,$ $U_g U_h = U_h U_g$ if
and only if $\omega (g,h) U_{gh} = \omega (h,g) U_{hg}$ if and only if $gh=hg$
and $\omega(g,h)=\omega(h,g).$ So, this condition is further equivalent to
$(g,1)(h,1) = (h,1)(g,1),$ which, in turn is equivalent to $(g, \alpha)(h,
\beta)=(h, \beta)(g,\alpha)$ for $\alpha, \beta \in \mathbf{T}_{\omega}$ and
that, in turn is equivalent to $(g,\alpha)(h, \beta)=(h,\beta)(g, \alpha)$ for
some
$\alpha, \beta \in \mathbf{T}_{\omega}.$ Thus, we have the following immediate
consequences of (v) above.\\
(a) $M \subset \mathbf{U}$ is an AUS if and only if $\mathbf{H}_M= \{(g,
\alpha) : U_g \in M, \alpha \in \mathbf{T}_{\omega}\}$ is a commutative subset
of $\mathbf{H}.$\\
(b) $M$ is a $\mathbf{U}$-MASS if and only if $\mathbf{H}_M$ is a maximal
abelian subgroup of $\mathbf{H}.$\\
(c) Put $G_M = \{g \in G : U_g \in M\}$ and consider any function $\rho$ on
$G_M$ to $\mathbf{T}_{\omega}.$ Set $T_{\rho} = \{(g, \rho(g)): g \in G_M \},$
the  $\rho$-transversal. We note that $G_M$ is the first projection of any such
$T_{\rho}$ as also of $\mathbf{H}_M.$\\
(d) Thus, the problem of finding MASS's in $\mathbf{U}$ is
equivalent
to that of finding maximal  abelian subgroups of $\mathbf{H}$ with different
first
projections.\\
(e) Further development of the theory of projective representations of finite
groups
studied thoroughly by I. Schur in early 1900s is very vast and deep. The survey
article
by Costache \cite{tlc} gives a readable account.  We will
not go into details or utilise or cite scholarly papers and monographs in this
paper.
\vskip0.1in
(vii)  Klappenecker and Roetteler \cite{ar} studied the following
question
of Schlingemann and Werner: Is every nice error basis (phase-) equivalent to a
basis of shift-and-multiply type? They answered it in the negative by concrete
examples using the theory of Heisenberg groups, theory of characters and
projective representations of finite groups. One can attempt alternate proofs 
using our results and details from the theory of finite groups.

\section{Applications to Quantum tomography}
\subsection{Motivation} 

Quantum tomography is the study of identification of quantum states by means of
a pre-assigned set of measurements. This set is usually taken to be a positive
operator-valued measure (POVM) viz., a set $\mathbf{A} = \{A_j : 1 \leq j
\leq v\}$ of positive operators on $\mathcal{H}$ with $\sum\limits_{j=1}^{v}
A_j = I.$ The quantum state $\rho$ on $\mathcal{H}$ is then attempted to be
determined via the tuple $\mathbf{\beta} = (\beta_j = (\tr \,\,(\rho
A_j))^v_{j=1}$ of
measurements. Because $\tr \,\,\rho = 1,$ we see that for any $j_0,$
$\beta_{j_{0}} = 1- \sum\limits_{j_{0} \neq j=1}^{v} \beta_j,$ and thus,
only
$v-1$ measurements are needed. If we can determine all states $\rho$ on
$\mathcal{H},$ then $\mathbf{A}$  is said to be informationally complete. For
that $v$ has to be $d^2$ or more. Without going into details which one can see,
for instance, in sources (\cite{Iv}, \cite{WF}, \cite{BBRV}, \cite{LBZ}) already referred to together with
the fundamental work on Quantum designs by Zauner \cite{gz} or recent papers like
\cite{rksc}, \cite{sg}, \cite{mww}, \cite{hmw}, \cite{gg} and \cite{aer}, we come straight to the case when
$\mathbf{A}$ is informationally complete and  all
$A_j$'s except possibly one have rank one. We call them {\it pure} POVMs. The
question
as to how our results help in constructing such a POVM of optimal size was
asked by K. R.
Parthasarathy. We thank him for that and also his motivating discussion on the
topic.

\subsection{Discussion}

Theoretically speaking, if $\mathbf{U}$ is a unitary basis then the tuple
$(\tr \,\, (\rho U_x))_{x \in X}$ determines the state $\rho.$ The same is true
if we take any unitary system $\mathbf{W} = \{W_y : y \in Y\}$ of size $(d^2-1)$ and consider the
tuple $(\tr \,\,(\rho W_y))_{y \in Y}$ as a representative of $\rho.$ We draw
upon [9], particularly excerpts given in Remark 2.3 (v).\\
(i) If  $\mathbf{W}$ can be partitioned
as union of $(d+1)$
 $\mathbf{W}$-MASS's, say $\{\mathbf{V}_s : 0 \leq s \leq d\}$ of size
$(d-1),$ then the complete system of $(d+1)$ orthonormal bases, say, $\left \{
\{\mathbf{b}_t^s : 0 \leq t \leq d-1\} , 0 \leq s \leq d \right \},$ of
$\mathbb{C}^d$ obtained via [9], Theorem 3.2 (viz., $\{\mathbf{b}_t^s : 0 \leq
t \leq d-1 \}$ is a common orthonormal eigenbasis for $\mathbf{\mathbf{V}_s},$
$0 \leq s \leq d$) gives rise to $(d^2-1)$ pure states
$\{\rho_j : 1 \leq j \leq d^2-1\}$ with $\rho_{s(d-1) +t}$ determined by the
unit vector $\mathbf{b}_t^s,$ $1 \leq t \leq d-1,$ $0 \leq s \leq d.$ Thus we
obtain a pure POVM $\{ A_j = \frac{1}{(d+1)} \rho_j, \,\,0 \leq j \leq d^2-1
\},$ where $\rho_0 = (d+1) I - \sum\limits_{j=1}^{d{^2}-1} \rho_j.$
\vskip0.1in
(ii) We further recall from \cite{BBRV} (see Remark 2.3(v)) that $\mathbf{W}$'s as in
(i) above do exist when $d$ is a prime power and note that Example 3.2 and
Figure 4 illustrate the situation for $4=2^2.$ In line with the contents of \cite{psmw} 
we make the following observations facilitated by Figure 4.

(a) Three navy blue circles put together overlap with the remaining subsystems.

(b) Three maroon lines put together overlap with the remaining subsystems. 

(c )Three yellow quadrilaterals put together can be combined with the middle and the outer circle. 
But, they overlap with green and sky-blue on the hexagon and also with maroon and navy-blue on the inner circle.

(d) Three sky-blue quadrilaterals put together can be combined with the middle and the inner circle. 
But, they overlap with green and yellow on the hexagon and also with maroon and navy-blue on the outer circle.

(e)  Three green triangles put together can be combined with the inner and the outer circle. 
But, they overlap with yellow and sky-blue on the hexagon and also with maroon and navy-blue on the middle circle.
\vskip0.1in
(iii) As noted in Theorem 2.8 (i) above, if each $W_y$ in $\mathbf{W}$ has
simple eigenvalues, then $\mathbf{V}_{\alpha}$'s constituting the fan
representation of $\mathbf{W}$ are mutually disjoint. This is a situation
similar to (i) above except that there is no guarantee that the sizes of
$\mathbf{V}_{\alpha}$'s are all $(d-1)$ or, equivalently, that of $\Lambda$ is
$(d+1).$ We do not yet have an example for that. In any case, the technique
indicated in (i) above does give a pure POVM of size $(d-1) \,|\Lambda| + 1.$
\vskip0.1in
(iv) In general, for a composite $d$ which is not a prime power, or, even when
$d$ is a prime power, a given $\mathbf{W}$ may not be decomposable as in (i)
above.
This can be seen in Example 3.1 (v) (e) \& Example 3.1 (vii) (e), and Example 3.1 (vii)(c) above
respectively. Figures 1 and 3 make it clear that the whole fan is needed to cover $\mathbf{U}$ 
and in Figure 2, only the red part can be ignored to obtain a smaller subset of the 
fan to cover $\mathbf{W}$.  Once again, Theorem 2.8(i) and Remark 2.9(iv) tell us that
overlapping $W_y$'s have to possess multiple eigenvalues. In the rest of this
section we make an attempt to obtain pure POVM's of optimal size for such
situations.
\vskip0.1in
(v) What comes in handy for our purpose is a minimal subset, say,
$\mathcal{M}_{\mathbf{W}},$ of  $\mathcal{V}_{\mathbf{W}}$ satisfying 
$\mathbf{W}= \cup \{ \mathbf{V}
: \mathbf{V} \in  \mathcal{M}_{\mathbf{W}}\}.$ We may write 
$\mathcal{M}_{\mathbf{W}} =
\{\mathbf{V}_{\alpha} : \alpha \in \Lambda_1\}$ with $\Lambda_1 \subset
\Lambda,$ if we like. A crude way to obtain a pure POVM would be to consider a
common orthonormal 
eigenbasis  $\{\mathbf{b}_t^{\alpha} : 0 \leq t \leq d-1\}$ for 
$\mathbf{V}_{\alpha}$ with $\alpha \in
\Lambda_1$ and construct a pure POVM as in (i) above, say $\mathbf{A} =
\{A_j : 0 \leq j \leq (d-1) |\Lambda_1|\}.$ We can refine this construction to
obtain a pure POVM of smaller size. We illustrate this refining process by means
of examples.

\subsection{Illustration} The context here is of Example 3.1 (vii). 
\vskip0.1in
(i) In view of items (a) and (b) there and minimality of
$\mathcal{M}_{\mathbf{W}},$ a crude bound $s_d$ for
$|\mathcal{M}_{\mathbf{W}}|$ can be given as follows.

For odd $d,$ $s_d = 4 + \frac{1}{2} \left [(d^2-1) - 4 (d-1) \right ] =
\frac{1}{2} \left [7 + (d-2)^2 \right ].$ For even $d > 2, \,\,s_d =4 +
\frac{1}{2}\left [(d^2-1) -
4 (d-1)+1 \right ] = 4 + \frac{1}{2}(d-2)^2.$ It is interesting that
for $d=4$ and $6,$ Example 3.1 (vii)(c) and (e) show that  $s_d$ is
$|\mathcal{M}_{\mathbf{W}}|.$ So by 4.2 (v) above, we can have a pure POVM of
size $\leq (d-1) s_d + 1.$ We can improve upon this for certain even $d$ as we
show in parts that follow.

(ii) Let us consider Example 3.1 (vii)(c) together with Figure 2. The unitary
system $\mathbf{W}$ is the union of the first six $\mathbf{W}$-MASS's. Each of
$(2,0),$  $(0,2),$ and $(2,2)$ occurs in one pair. For instance,
$(2,2)$ occurs in \linebreak
$\left \{\left \{(2,2), (1,1), (3,3)\right \} \right \},$ $\left \{\left
\{(2,2), (3,1), (1,3) \right \}\right \}.$ Each of $(2,0),$ $(0,2),$ $(2,2)$
has eigenvalues $1$ and $-1$ both of multiplicity two. We pick up any of these
three unitaries, say, $(2,2),$ and the corresponding pair of $\mathbf{W}$-MASS's
for illustration. Let $\mathcal{H}_1$ and $\mathcal{H}_2$ be the two
corresponding eigenspaces so that $\mathbb{C}^4 = \mathcal{H}_1 \oplus
\mathcal{H}_2.$ We can write the three unitaries $(2,2),$ $1,1),$ $(3,3)$ as
blocks of operators $\left [\begin{array}{cc} I & 0 \\ 0 & -I \end{array} \right
],$ $\left [\begin{array}{cc}R_1 & 0 \\ 0 & R_2 \end{array} \right ],$ and
$\left [\begin{array}{cc} S_1 & 0 \\ 0 & S_2 \end{array} \right ]$ respectively.
Let $\left \{\xi_j^k, k,j = 1,2 \right \}$ be a  common system of orthonormal
eigenvectors with $\xi_j^k \in \mathcal{H}_k$ for  $k,j = 1,2.$ Then  $\rho_{
\xi^{k}_{1}} + \rho_{\xi_{2}^{k}} = P_{\mathcal{H}_{k}},$ the projection on
$\mathcal{H}_k$
for $k=1,2.$  Similar remarks apply  to the unitaries $(2,2),$ $(3,1),$ $(1,3)$
and we take the common system of orthonormal eigenvectors as $\{\eta_j^k : k, j
= 1,2 \}.$ Then $\rho_{\eta_{1}^{k}} + \rho_{\eta_{2}^{k}} =
P_{\mathcal{H}_{k}}$ for $k=1,2.$ We may take the unit vectors $\{\xi_1^1,
\xi_2^1, \xi_1^2,  \eta_1^1, \eta_1^2 \}$ and the corresponding pure
states $\{ \rho_1^1, \rho_2^1, \rho_1^2, \sigma_1^1, \sigma_1^2\}.$ We collect
such pure state five-tuples for the remaining two pairs of
$\mathbf{W}$-MASS's. Thus, we have exactly 15 pure states which give rise to a
desired pure POVM.
\vskip0.1in
(iii) Let us consider Example 3.1 (vii)(e) now together with Figure 3. We need
all the 12 $\mathbf{W}$-MASS's to make up $\mathbf{W}$ as their union. Each of
the
move-together pairs  $\{(2,0), (4,0)\}, $ $\{(2,2), (4,4)\}, $ $\{(2,4),
(4,2)\},
$ and $\{(0,2), (0,4)\}$ occurs in exactly three $\mathbf{W}$-MASS's. The
move togethers  are scalar multiples of inverses of each other. So they have
common eigenvectors, and eigenvalues, though different as tuples but 
of the same multiplicity 2. We pick up any, say, $(2,2), (4,4)$ for
illustration. So
we
can write $\mathbb{C}^6$ as $\mathcal{H}_1 \oplus \mathcal{H}_2 \oplus
\mathcal{H}_3,$ with each $\mathcal{H}_k$ as eigenspace for $(2,2)$
corresponding to eigenvalue, say, $\lambda_k$ of multiplicity 2. We can write
the
operators in the first $\mathbf{W}$-MASS, say, $Y_{-4}$ as block  $\left [
\begin{array}{ccc} R_1 & 0 & 0 \\ 0 & R_2 & 0 \\ 0 & 0 & R_3  \end{array} \right
],$  $\left [
\begin{array}{ccc} S_1 & 0 & 0 \\ 0 & S_2 & 0 \\ 0 & 0 & S_3  \end{array} \right
],$ $\left [
\begin{array}{ccc} \lambda_1 I & 0 & 0 \\ 0&\lambda_2 I & 0 \\ 0 & 0& \lambda_3
I 
\end{array} \right ],$ $\left [
\begin{array}{ccc} \mu_1, I & 0 & 0 \\ 0 & \mu_2 I & 0 \\ 0 & 0 & \mu_3 I 
\end{array} \right ],$ $\left [
\begin{array}{ccc} T_1 & 0 & 0 \\ 0 & T_2 & 0 \\ 0 & 0 & T_3 
\end{array} \right ]$ with $S_k = \mu R_k^{-1},$ for some
$\mu \in S^1.$\\ Let $\left \{\xi_j^k, j = 1,2, k=1,2,3 \right
\}$ be a
common system of eigenvectors with $\xi_j^k \in \mathcal{H}_k$ for $k=1,2,3$
and $j=1,2.$
Arguing as in (i) above we can have common systems of eigenvectors $\left \{
\eta_j^k : j =1,2; k = 1,2,3 \right \}$ and $\left \{\zeta_j^k : j =1,2; k 
\right . $ $\left . = 1,2,3 \right \}$ for the remaining two
$\mathbf{W}$-MASS's. We may take the unit vectors \linebreak $\left \{ \xi_1^1,
\xi_2^1,
\xi_1^2, \xi_2^2, \xi_1^3, \eta_1^1, \eta_1^2, \eta_1^3, \zeta_1^1, \zeta_1^2,
\zeta_1^3 \right \}$ and the corresponding pure states \linebreak $\left
\{\rho_1^1,
\rho_2^1, \rho_1^2, \rho_2^2, \rho_1^3, \sigma_1^1, \sigma_1^2, \sigma_1^3,
\gamma_1^1, \gamma_1^2, \gamma_1^3 \right \}.$ We collect such pure state
11-tuples for the remaining 3 triples of $\mathbf{W}$-MASS's Thus we have a pure
POVM of size 45. The ideal desired number is 36, of course but, this method does
not give this, if any. We note that doing it via $(3,0),$ $(0,3),$ $(3,3)$ we
will
obtain a pure POVM of size 52 by our method.
\vskip0.1in
(iv) Let us consider the case $d=2r$ with $r$ odd, $\ge 3$ in Example 3.1 (vii).
As
noted in part (b) of Example 3.1(vii), $(r,0) \Delta (m^{\prime}, n^{\prime})$
if and only if $n^{\prime}$
is even, $(0,r) \Delta (m^{\prime}, n^{\prime})$ if and only if $m^{\prime}$ is
even, $(r,r)\Delta(m^{\prime}, n^{\prime})$ if and only if
$m^{\prime}-n^{\prime}$ is even. Also $(r,0),$ $(0,r),$ $(r,r)$ do not commute
with each other. So a minimal $\mathcal{M}_{\mathbf{W}}$ exists satisfying:
$(r,0)$
present in $u$ $\mathbf{W}$-MASS's in $\mathcal{M}_{\mathbf{W}},$ $(0,r)$ in
$v$ of  them and $(r,r)$ in $w$ of them, and, these three types exhaust
$\mathcal{M}_{\mathbf{W}}.$ Then $u+v+w=|\mathcal{M}_{\mathbf{W}}|.$  The
technique indicated for $r=3$ in (iii) above gives that there exists a pure
POVM of size $\leq ((d-1) + (u-1) (d-2)) + ((d-1) + (v-1)(d-2))+ ((d-1) +
(w-1)(d-2))+1 = 4 + (d-2) (u+v+w)= 4 + (d-2) |\mathcal{M}_{\mathbf{W}}|.$ This
agrees with the number 52 for $r=3,$ for in that case $u=v=w=4$ and
$\mathcal{M}_{\mathbf{W}}|=12.$ In the general case, the crude bound $s_d$
given in 4.3(i) can be used. Thus we have a pure POVM of size $\leq 4 + (d-2)
(4+\frac{1}{2} (d-2)^2) = 4 [1 + 2 (r-1) +(r-1)^3].$
\vskip0.1in

\section{Conclusion}
We started with a unitary basis $\mathbf{U}=\{U_x, x \in  X\}$ on
a Hilbert space $\mathcal{H}$ of dimension $d$ and an $x_0 \in X.$ We
associated the tag  $\mathbf{T}=(x_0, U_{x_{0}}, \mathbf{W})$ at
$x_0,$ where  $\mathbf{W} = \{W_x = U_{x_{0}}^{\ast} U_x, x \in X,
x \neq x_0   \}.$ We obtained a covering of
$\mathbf{W}$ by maximal abelian subsets of $\mathbf{W}$ (called $\mathbf{W}$
MASS's). We obtained the set of $\mathbf{W}$ MASS's for different concrete
$\mathbf{W}$'s displaying various patterns, like mutually disjoint, overlapping
in different ways, and, therefore, called them fans. Varying $x_0,$ the whole
collection was called a fan system of $\mathbf{U}.$ We showed that it is an
invariant of $\mathbf{U}$ to within (phase) equivalence of unitary bases \cite{Wer}.
The concept of
collective unitary equivalence was utilised for this purpose. It was also used
to construct positive partial transpose matrices. Finally, applications of fan
to quantum tomography were indicated. Examples have been given to illustrate the results.

\section*{Acknowledgement}

The second author expresses her deep sense of gratitude to \linebreak K.
R. Parthasarathy. She has learnt most of the basic concepts in this
paper from him during his Seminar Series of Stat. Math. Unit at the Indian
Statistical Institute, New Delhi, University of Delhi and elsewhere. She has
gained immensely from insightful discussion sessions with him
from time to time.
\vskip0.1in
She thanks V. S. Sunder and V. Kodiyalam for supporting her visit to The
Institute of Mathematical Sciences (IMSc), Chennai to participate in their
scholarly workshops on ``Functional Analysis of Quantum Information Theory'' in
December, 2011 - January, 2012, ``Planar algebras'' in March-April, 2012 and
Sunder Fest in April, 2012. This enabled her to learn more from experts at these
events and initiate her interaction with her co-author. She also thanks R.
Balasubramanian,
Director, IMSc, for providing more such opportunities, kind hospitality,
stimulating research atmosphere and encouragement
all through. 
\vskip0.1in
She  thanks Indian National Science Academy for support under the INSA Senior
Scientist and Honorary Scientist  Programmes and Indian Statistical
Institute, New Delhi for  Visiting positions under these
programmes
together with excellent research atmosphere and facilities all the time.
\vskip0.1in
Finally, the authors thank Kenneth A. Ross for reading the paper and suggesting
improvements. They also thank Mr. Anil Kumar Shukla for transforming the
manuscript
into its present \LaTeX \,\,form, and for his cooperation and patience
during the
process.

\end{document}